\documentclass{gGAF2e}
\usepackage{natbib}
\usepackage{color}
\usepackage{makecell}
\pdfoutput=1
\newcommand{\rev}[1]{{\color{blue}#1}} 
\usepackage[normalem]{ulem}
\usepackage[center]{caption}

\begin{document}

\jvol{00} \jnum{00} \jyear{2023} 

\markboth{\rm C.~LIU and  A.D.~CLARK}{\rm GEOPHYSICAL \&  ASTROPHYSICAL FLUID DYNAMICS}

\title{Semi-analytical solutions of shallow water waves\\  with idealised bottom topographies}

\author{Chang Liu ${\dag}$ ${\ddag}$ $^\ast$\thanks{$^\ast$Corresponding author. Email: cliu124@alumni.jh.edu
\vspace{6pt}} and Antwan~D.~Clark ${\dag}$ \\\vspace{6pt} ${\dag}$ Department of Applied Mathematics and Statistics, Johns Hopkins University \\ Baltimore,
MD 21218 USA\\ ${\ddag}$ Department of Physics, University of California, Berkeley \\
Berkeley, CA 94720 USA  \\\vspace{6pt}\received{v4.4 released July 2022} }

\maketitle

\begin{abstract}
	
Analysing two-dimensional shallow water equations with idealised bottom topographies have many applications in the atmospheric and oceanic sciences; however, restrictive flow pattern assumptions have been made to achieve explicit solutions. This work employs the Adomian decomposition method (ADM) to develop semi-analytical formulations of these problems that preserve the direct correlation of the physical parameters while capturing the nonlinear phenomenon. Furthermore, we exploit these techniques as reverse engineering mechanisms to develop key connections between some prevalent ansatz formulations in the open literature as well as derive new families of exact solutions describing geostrophic inertial oscillations and anticyclonic vortices with finite escape times. Our semi-analytical evaluations show the promise of this approach in terms of  providing robust approximations against several oceanic variations and bottom topographies while also preserving the direct correlation between the physical parameters such as the Froude number, the bottom topography, the Coriolis parameter, as well as the flow and free surface behaviours. Our numerical validations provide additional confirmations of this approach while also illustrating that ADM can also be used to provide insight and deduce novel solutions that have not been explored, which can be used to characterize various types of  geophysical flows. 

\end{abstract}

\section{Introduction}
\par Analysing two-dimensional shallow water equations has been extensively studied in geophysical fluid dynamics to understand a myriad of atmospheric and oceanic phenomena. Some examples include understanding the effects of long-term oceanic waves \citep{pedlosky2013geophysical, vallis2017atmospheric}, analyzing the behaviour of oceanic warm-core rings \citep{cushman1987exact}, investigating flows in channels and shorelines \citep{shapiro1996nonlinear, sampson2005moving}, studying steady-state flows \citep{iacono2005analytic, sun2016high}, and grasping the temporal instability of barotropic zonal flows \citep{clark2013improved}. These theoretical analyses also serve as a good basis for numerical simulations and validations. For example, the creators of the Shallow Water Analytic Solutions for Hydraulic and Environmental Studies (SWASHES) software library \citep{delestre2013swashes} incorporated a significant number of theoretical solutions of the shallow water equations in the open literature, which has been cited by over $200$ research papers currently. Furthermore, several of the solutions in this library are obtained from \citet{thacker1981some} in which have been widely used to demonstrate the validity and accuracy of several numerical schemes including finite volume schemes \citep{gallardo2007well,bollermann2011finite,nikolos2009unstructured} and discontinuous Galerkin methods \citep{ern2008well,kesserwani2012locally,li2017positivity,wintermeyer2018entropy}. Some significant advancements include the original works of Ball and Thacker who demonstrated that nonlinear oscillations can be modelled as either low-order polynomials or normal modes \citep{ball1963some, ball1964exact, ball1965effect, thacker1977irregular, thacker1981some}. Researchers also developed elliptical vortex solutions to understand the temporal effects of oceanic warm-core rings including stationary clockwise rotations (rodons), pulsating circular eddies (pulsons), and a subclass of these phenomena called pulsrodons \citep{cushman1987exact, cushman1985oscillations, rogers1989elliptic}. Extensions to these approaches have been made, where some examples include the work of
\cite{sachdev1996regular} who extended the approach of  \cite{clarkson1989new} and derived new families of solutions in paraboloidal basins that provided additional insights in terms of describing flow behaviour due to deformation modes. Additionally, \citet{matskevich2019exact} extended the results of Ball and Thacker to include the effects of Coriolis forces and bottom friction.
\cite{bristeau2021analytical} also extended the results of Thacker and introduced two respective solutions describing velocity distributed along the vertical axis and velocity accounting for variable density. 

\par Group analysis was also explored. Some pioneering works in this area include that of
\cite{curro1989some} and
\cite{rogers1989generation} who also advanced the works of Thacker and Ball and related several forms of the depth function as well as developed invariance theorems.
\cite{levi1989group} developed symmetry reductions for flows with elliptic and circular bottom topographies.
\cite{mansfieldsymmetry} derived Lie point symmetries and conservation laws. \citet{chesnokov2009symmetries} discovered $9$-dimensional Lie algebra point symmetries and developed transformations between rotating and non-rotating cases, which were later used to describe spatial oscillations in spinning paraboloids \citep{chesnokov2011properties}. Some recent advancements include \citet{meleshko2020complete} and
\cite{bihlo2020lie} who performed group classification and analysis for zero and constant Coriolis parameters. Meanwhile, \citet{meleshko2020group} performed similar analysis and considered the beta-plane approximation of the Coriolis parameter and an irregular bottom topography.

\par However, deriving theoretical solutions to the two-dimensional shallow water equations poses the following main challenges. First, these efforts involve making specific assumptions regarding the flow conditions which only satisfy specific cases. Some solutions also contain combinations of special functions and integral expressions \citep{shapiro1996nonlinear, rogers1989elliptic}, which in turn makes it difficult to determine the correlation between the physical quantities of these models. Finding invariant solutions via group analysis has the additional advantage of deriving conservation laws to these equations. However, this approach depends on the construction of Lie-groups which depend on the problem formulation as well as specific assumptions such as the Coriolis parameter and bottom topography. Therefore, there is a need to find solutions that are not only flexible, in terms of relaxing certain limiting assumptions, but also provide a direct correlation of the physical parameters.

\par This work applies Adomian decomposition method (ADM) \citep{adomian1990review} to the shallow water equations to provide the following main contributions. First, we present the ADM formulation of the rotating shallow water equations where we also present key connections between the ansatz formulations in the work of \cite{thacker1981some,shapiro1996nonlinear,matskevich2019exact}. Next, we derive and present some new families of exact solutions, for flat bottom topographies, that describe inertial oscillations in geostrophic flows and anticyclonic vortices with finite escape times. The rest of this paper is organised in the following manner. Section \ref{sec:SWE} presents the ADM formulation and initial theoretical formulation of the problem, where we present the connection to fundamental assumptions on the formulation of the solutions. Section \ref{sec:new_exact} presents derivations of new families of solutions and their properties. Section \ref{sec:numerical} provides numerical experimentation and results. Section \ref{sec:discussion} provides some concluding remarks, where we also list some future research directions. 

\section{Adomian Decomposition Formulation}

\label{sec:SWE}
\begin{figure}[h!]
    \centering
    \includegraphics[width=6in]{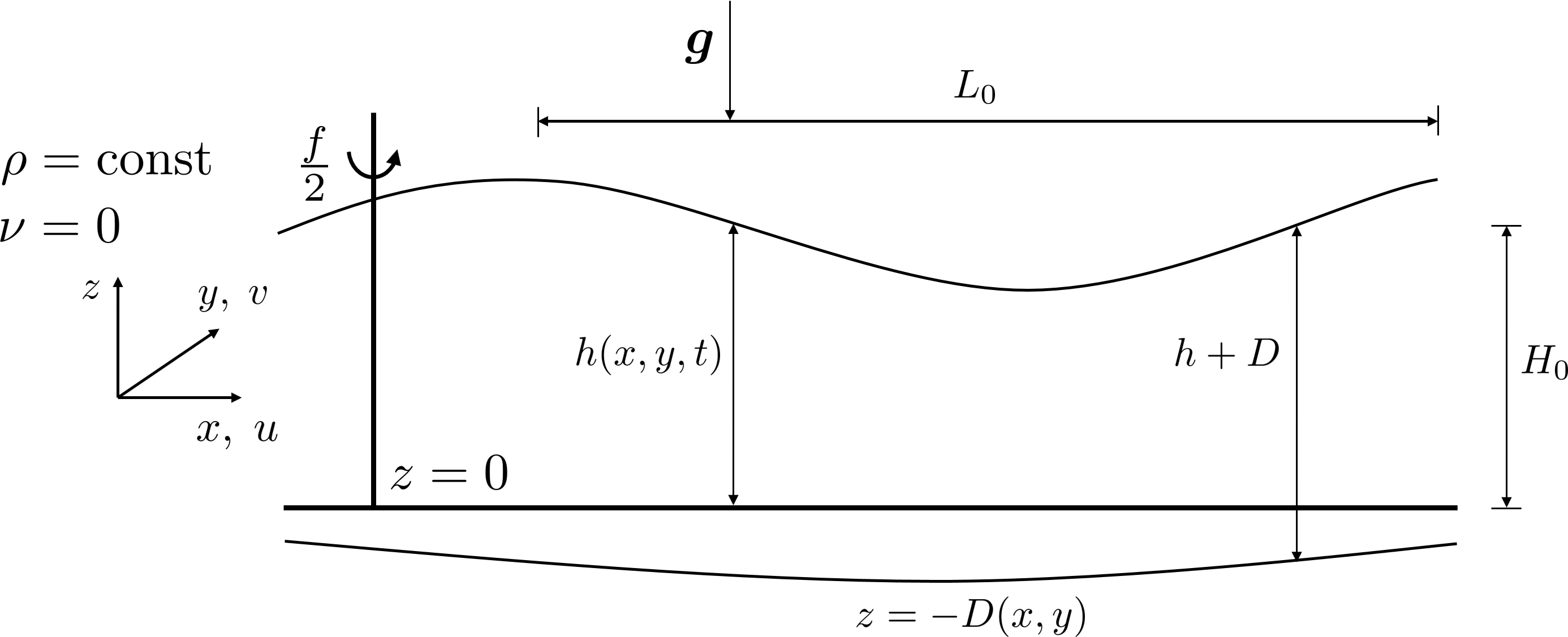}
    \caption{Illustration of a thin layer of incompressible flow under the Earth's rotation described by rotating shallow-water equations with idealised bottom topography.}
\label{fig:SWE}
\end{figure}

\par The non-dimensional form of the governing equations is defined as
\begin{subequations}
\label{eq:SWE_non_dimen}
\begin{align}
    \frac{\upartial u}{\upartial t}=&-u\frac{\upartial u}{\upartial x}-v\frac{\upartial u}{\upartial y}-\frac{1}{F^2}\frac{\upartial h}{\upartial x}+\bar{f}v,  \\
    \frac{\upartial v}{\upartial t}=&-u\frac{\upartial v}{\upartial x}-v\frac{\upartial v}{\upartial y}-\frac{1}{F^2}\frac{\upartial h}{\upartial y}-\bar{f}u, \\
    \frac{\upartial h}{\upartial t}=&-\frac{\upartial }{\upartial x}[u(h+D)]-\frac{\upartial }{\upartial y}[v(h+D)].
\end{align}
\end{subequations}
This is illustrated in figure \ref{fig:SWE}, where $u$ and $v$ are the flow velocity components, $h$ is the free surface height, $\bar{f}=fL_0/U_0$ is the dimensionless Coriolis parameter (associated with the Coriolis force), and $F = U_0/\sqrt{gH_0}$ is the Froude number. Here, the spatial variables $x$, $y$, $l$, and $L$ are normalised by the horizontal length scale $L_0$; $h$ is normalised by a vertical length scale $H_0$; the horizontal velocities, $u$ and $v$, are normalised by the characteristic velocity $U_0$; and time $t$ is normalised by $L_0/U_0$. Hence, the dimensionless form of the idealised bottom topography is defined as
\begin{equation}
    D(x,y) = D_0 \left(1-\frac{x^2}{L^2}-\frac{y^2}{l^2}\right),
\label{eq:idealBottomForm}
\end{equation}
where $D_0$ is also normalised by a vertical length scale $H_0$. It is noteworthy to mention that other bottom topographies can be determined from \eqref{eq:idealBottomForm} such as  flat bottom ($D_0=0$), circular paraboloid ($l=L$), and channel ($l\rightarrow \infty$ or $L\rightarrow \infty$) terrains. Additionally, $D(x,y)$ can also be used to incorporate linear terms in its description via change of variables \citep{shapiro1996nonlinear,thacker1981some}. The total fluid depth $D+h$, shown in figure \ref{fig:SWE}, follows the formulations of \cite{thacker1981some} and \cite{shapiro1996nonlinear} where  $D+h=0$ represents a moving shoreline and $D+h < 0$ represents dry regions. When the moving shoreline is closed, the water mass within the shoreline is conserved  \citep{thacker1981some,shapiro1996nonlinear}. When the moving shoreline is open such as in tsunami modelling, then water within a bounded domain will have mass exchange with an infinite mass reservoir. It is also important to mention that our explorations in this section consider flow velocities that are linearly varying spatially while the free surface height either varies linearly or in a quadratic fashion. The initial conditions are given by
\begin{subequations}
\label{eq:SWE_initial}    
\begin{align}
    u(x,y,0) =& u_0 (x,y),\\
    v(x,y,0) =& v_0(x,y),\\
    h(x,y,0) =& h_0(x,y).
\end{align}
\end{subequations}

Next, $u$, $v$, and $h$ are decomposed as follows
\begin{align}
    \begin{bmatrix}
    u(x,y,t)\\
    v(x,y,t)\\
    h(x,y,t)
    \end{bmatrix}=\sum_{n=0}^\infty\begin{bmatrix}
    u_n(x,y,t)\\
    v_n(x,y,t)\\
    h_n(x,y,t)
    \end{bmatrix},
\label{eq:ADM_series}    
\end{align}
where the initial components are defined by equation \eqref{eq:SWE_initial}. 
Thus, the recurrence relationships to equation \eqref{eq:SWE_non_dimen} (for $n \geq 0$) are given by 
\begin{subequations}
    \label{eq:ADM_iter_Linv}
\begin{align}
    &u_{n+1}(x,y,t)={L}_t^{-1}\left\{-A_n\left(u,\frac{\upartial u}{\upartial x}\right)-A_n\left(v,\frac{\upartial u}{\upartial y}\right)-\frac{1}{F^2}\frac{\upartial h_{n}}{\upartial x}+\bar{f}v_n\right\},  \\
    &v_{n+1}(x,y,t)={L}_t^{-1}\left\{-A_n\left(u,\frac{\upartial v}{\upartial x}\right)-A_n\left(v,\frac{\upartial v}{\upartial y}\right)-\frac{1}{F^2}\frac{\upartial h_{n}}{\upartial y}-\bar{f}u_n\right\}, \\
    &h_{n+1}(x,y,t)={L}_t^{-1}\left\{-\frac{\upartial }{\upartial x}[A_n(u,h)]-\frac{\upartial }{\upartial y}[A_n(v,h)]-\frac{\upartial }{\upartial x}[u_n D]-\frac{\upartial }{\upartial y}[v_n D]\right\},
\end{align}
\end{subequations}
where
\begin{align}
    L_t =\frac{\upartial (\cdot)}{\upartial t},\;\;
    \text{and}\;\;
    L_t^{-1}=\int_{0}^{t}(\cdot) \; {\rm d} \tau, \nonumber
\end{align}
and the Adomian polynomial representing the quadratic nonlinearity is defined as \citep{adomian1990review,adomian2013solving} 
\begin{align}
    A_n(u,h) =\sum_{j=0}^n u_{j}h_{n-j}.
\label{eq:ADM_An}
\end{align}
It is important to note that equation \eqref{eq:ADM_An} can be used to approximate the quadratic nonlinear terms, such as $uh$, as follows
\begin{align}
    uh=\left(\sum_{p}^{\infty}u_p \right)\left(\sum_{q}^{\infty}h_q \right)=\sum_{n}^{\infty}A_n(u,h)\nonumber
\end{align}
and thus the semi-analytical solution to \eqref{eq:SWE_non_dimen} is expressed via the partial sums
\begin{align}
    \begin{bmatrix}
    u(x,y,t)\\
    v(x,y,t)\\
    h(x,y,t)
    \end{bmatrix}=\begin{bmatrix}
        S_N(u)\\
        S_N(v)\\
        S_N(h)
    \end{bmatrix}=\sum_{n=0}^N \begin{bmatrix}
        u_n(x,y,t)\\
        v_n(x,y,t)\\
        h_n(x,y,t)
    \end{bmatrix}.
    \label{eq:partial_sum}
\end{align}
Next, the following results connect the properties of the initial conditions to the behaviours of the true solutions via their partial sums. 

\begin{lemma}
Let $\{u_n(x,y,t)\}$, $\{v_n(x,y,t)\}$, $\{h_n(x,y,t)\}$ be the sequence of decomposed functions of $u$, $v$, and $h$, where their relationship is defined by \eqref{eq:ADM_iter_Linv} (for $n \in \mathbb{N}$) given an ideal parabolic topography \eqref{eq:idealBottomForm}. If the initial conditions $u_0(x,y)$, $v_0(x,y)$, $h_0(x,y)$ are defined such that 
\begin{subequations}
    \begin{align}
     \frac{\upartial^2 u_0(x,y)}{\upartial x^2} =& \frac{\upartial^2 v_0(x,y)}{\upartial x^2} = \frac{\upartial^3 h_0(x,y)}{\upartial x^3} = 0,
\label{PartialDerivIC_x_Cond}\\
 \frac{\upartial^2 u_0(x,y)}{\upartial y^2} = &\frac{\upartial^2 v_0(x,y)}{\upartial y^2} = \frac{\upartial^3 h_0(x,y)}{\upartial y^3} = 0,
\label{PartialDerivIC_y_Cond}\\
  \frac{\upartial^2 u_0(x,y)}{\upartial xy}=&\frac{\upartial^2 v_0(x,y)}{\upartial x \upartial y}= \frac{\upartial^3 h_0(x,y)}{\upartial x^2 \upartial y} = \frac{\upartial^3 h_0(x,y)}{\upartial x\upartial y^2} = 0.
\label{PartialDerivIC_mixedxy_Cond}
    \end{align}
\end{subequations}
Then the higher order components $u_n(x,y,t)$, $v_n(x,y,t)$, $h_n(x,y,t)$ also satisfy the same property, where 
\begin{subequations}
    \begin{align}
        \frac{\upartial^2 u_n(x,y,t)}{\upartial x^2} =& \frac{\upartial^2 v_n(x,y,t)}{\upartial x^2} = \frac{\upartial^3 h_n(x,y,t)}{\upartial x^3} = 0,
\label{PartialDeriv_x_Cond}\\
        \frac{\upartial^2 u_n(x,y,t)}{\upartial y^2} =& \frac{\upartial^2 v_n(x,y,t)}{\upartial y^2} = \frac{\upartial^3 h_n(x,y,t)}{\upartial y^3} = 0,
\label{PartialDeriv_y_Cond}\\
        \frac{\upartial^2 u_n(x,y,t)}{\upartial x\upartial y}=&\frac{\upartial^2 v_n(x,y,t)}{\upartial x\upartial y}= \frac{\upartial^3 h_n(x,y,t)}{\upartial x^2\upartial y} = \frac{\upartial^3 h_n(x,y,t)}{\upartial x\upartial y^2} = 0
\label{PartialDeriv_mixedxy_Cond}
    \end{align}
\end{subequations}
for $n \in \mathbb{N}^{+}$.
\label{ADM:PartialDerivativeCoeffLemma}
\end{lemma}

\begin{proof}
This is proven via mathematical induction by examining the recursion relationships for $u$, $v$, and $h$ in equation \eqref{eq:ADM_iter_Linv}. Condition \eqref{PartialDeriv_x_Cond} is demonstrated by examining the following relationships
\begin{subequations}
    \begin{align}
        \frac{\upartial^{2} u_{n+1}}{\upartial x^{2}}=&-{L}_t^{-1}\left\{\frac{\upartial^{2}}{\upartial x^{2}}\left[A_n\left(u,\frac{\upartial u}{\upartial x}\right) + A_n\left(v,\frac{\upartial u}{\upartial y}\right) \right] + \frac{1}{F^2}\frac{\upartial^{3} h_{n}}{\upartial x^{3}} - \bar{f}\, \frac{\upartial^{2}v_n}{\upartial x^{2}}\right\},
\label{un_xx}\\
\frac{\upartial^{2} v_{n+1}}{\upartial x^{2}} =&-{L}_t^{-1}\left\{\frac{\upartial^{2}}{\upartial x^{2}} \left[A_n\left(u,\frac{\upartial v}{\upartial x}\right) + A_n\left(v,\frac{\upartial v}{\upartial y}\right) \right] + \frac{1}{F^2}\frac{\upartial^{3} h_{n}}{\upartial x^{2}\upartial y} + \bar{f} \, \frac{\upartial^{2} u_n}{\upartial x^{2}} \right\},
\label{vn_xx}\\
\frac{\upartial^{3} h_{n+1}}{\upartial x^{3}} =& - {L}_t^{-1}\left\{\frac{\upartial^{4} }{\upartial x^{4}}[A_n(u,h)] + \frac{\upartial^{4} }{\upartial x^{3}\upartial y}[A_n(v,h)]+\frac{\upartial^{4} }{\upartial x^{4}}[u_n D] + \frac{\upartial^{4} }{\upartial x^{3}\upartial y}[v_n D]\right\}.\hskip 6mm
\label{hn_xxx}
    \end{align}
\end{subequations}
Therefore, when $n = 0$ equations  (\ref{un_xx}-c)
representing the relationship between the initial and first components for $u$, $v$, and $h$ become
\begin{subequations}
    \begin{align}
\frac{\upartial^{2} u_{1}}{\upartial x^{2}}=& -{L}_t^{-1}\left\{\frac{\upartial^{2}}{\upartial x^{2}}\left[u_{0}\frac{\upartial u_{0}}{\upartial x} \,  + v_{0}\frac{\upartial u_{0}}{\upartial y} \right] + \frac{1}{F^2}\frac{\upartial^{3} h_{0}}{\upartial x^{3}} - \bar{f}\, \frac{\upartial^{2}v_{0}}{\upartial x^{2}}\right\},
\label{u1_xx}\\
\frac{\upartial^{2} v_{1}}{\upartial x^{2}} =& -{L}_t^{-1}\left\{\frac{\upartial^{2}}{\upartial x^{2}} \left[u_{0} \frac{\upartial v_{0}}{\upartial x}+ v_{0} \frac{\upartial v_{0}}{\upartial y} \right] + \frac{1}{F^2}\frac{\upartial^{3} h_{0}}{\upartial x^{2} \upartial y} + \bar{f} \, \frac{\upartial^{2} u_{0}}{\upartial x^{2}} \right\},
\label{v1_xx}\\
\frac{\upartial^{3} h_{1}}{\upartial x^{3}} =& - {L}_t^{-1}\left\{\frac{\upartial^{4} }{\upartial x^{4}}\left[u_{0} h_{0} \right] + \frac{\upartial^{4} }{\upartial x^{3}y} \left[v_{0} h_{0} \right]+\frac{\upartial^{4} }{\upartial x^{4}} \left[u_{0} D \right] + \frac{\upartial^{4} }{\upartial x^{3}\upartial y} \left[v_{0} D \right]\right\}. \hskip 8mm
\label{h1_xxx}
    \end{align}
\end{subequations}
Employing (\ref{PartialDerivIC_x_Cond}-c)
it can be shown that equations (\ref{u1_xx}-c)
reduce to the following relationship
\begin{equation}
    \frac{\upartial^2 u_1(x,y,t)}{\upartial x^2} = \frac{\upartial^2 v_1(x,y,t)}{\upartial x^2} = \frac{\upartial^3 h_1(x,y,t)}{\upartial x^3} = 0.\nonumber
\label{PartialDerivComp1_x_Cond}
\end{equation}
\noindent Continuing this argument for $n = \{1, 2, \ldots, n - 1\}$ yields equation  \eqref{PartialDeriv_x_Cond}. Similar arguments can be made to produce (\ref{PartialDeriv_y_Cond},c), respectively.
\end{proof}

\begin{theorem}
\label{thm:ADM_Ansatz_Connection}
Let $\{u_n(x,y,t)\}$, $\{v_n(x,y,t)\}$, $\{h_n(x,y,t)\}$ be the sequence of decomposed functions of $u$, $v$, and $h$, where their relationship is defined by \eqref{eq:ADM_iter_Linv} (for $n \in \mathbb{N}$) given an ideal parabolic topography \eqref{eq:idealBottomForm}. If the initial conditions $u_0(x,y)$, $v_0(x,y)$, $h_0(x,y)$ are defined as
{\rm (\ref{PartialDerivIC_x_Cond}-c)},
then the solutions of $u$, $v$, and $h$ have the same property where 
\begin{subequations}
    \begin{align}
    \frac{\upartial^2 u(x,y,t)}{\upartial x^2} = &\frac{\upartial^2 u(x,y,t)}{\upartial y^2} 
    = \frac{\upartial^2 u(x,y,t)}{\upartial x \upartial y} = 0,
\label{PartialCondition_u}   \\
\frac{\upartial^2 v(x,y,t)}{\upartial x^2} =& \frac{\upartial^2 v(x,y,t)}{\upartial y^2} 
    = \frac{\upartial^2 v(x,y,t)}{\upartial x \upartial y} = 0,
\label{PartialCondition_v}\\
 \frac{\upartial^3 h(x,y,t)}{\upartial x^3}
    =& \frac{\upartial^3 h(x,y,t)}{\upartial x^2 \upartial y} 
    = \frac{\upartial^3 h(x,y,t)}{\upartial x \upartial y^2} 
    = \frac{\upartial^3 h(x,y,t)}{\upartial y^3} = 0.
\label{PartialCondition_h}
    \end{align}
\end{subequations}
Consequently, these solutions can be expressed as 
\begin{subequations}
    \begin{align}
         {u}(x,y,t)=&\tilde{u}_0(t)+\tilde{u}_x(t)x+\tilde{u}_y(t)y,
\label{uLinearCondition}\\
{v}(x,y,t)=&\tilde{v}_0(t)+\tilde{v}_x(t)x+\tilde{v}_y(t)y,
\label{vLinearCondition}\\
  {h}(x,y,t)=&\tilde{h}_0(t)+\tilde{h}_x(t)x+\tilde{h}_y(t)y+\tfrac{1}{2}\tilde{h}_{xx}(t)x^2+\tfrac{1}{2}\tilde{h}_{yy}(t)y^2+\tilde{h}_{xy}(t)xy, \hskip 10mm 
\label{hLinearCondition}
    \end{align}
\end{subequations}
where the coefficients $\tilde{u}_0(t)$, $\tilde{u}_x(t)$, $\tilde{u}_y(t)$, $\tilde{v}_0(t)$, $\tilde{v}_x(t)$, $\tilde{v}_y(t)$, $\tilde{h}_0(t)$, $\tilde{h}_x(t)$, $\tilde{h}_y(t)$, $\tilde{h}_{xx}(t)$, $\tilde{h}_{yy}(t)$, and $\tilde{h}_{xy}(t)$ are time-dependent.
\end{theorem}

\begin{proof}
  Applying Lemma \ref{ADM:PartialDerivativeCoeffLemma} to each component in \eqref{eq:ADM_series} yields
(\ref{PartialCondition_u}-c).
  From \eqref{PartialCondition_u}, we observe that 
\begin{align}
   \frac{\upartial^2 u(x,y,t)}{\upartial x^2} = 0 \hskip 5mm  \mbox{yields} \hskip 5mm  u(x,y,t)=C_1(y,t) x+C_2(y,t),\nonumber
\end{align}
where the integration constants, $C_1(y,t)$ and $C_2(y,t)$, are independent of $x$. Similarly, we have 
\begin{align}
   \frac{\upartial^2 u(x,y,t)}{\upartial x \upartial y} = 0  \hskip 5mm\mbox{yields}  \hskip 5mm  C_1(y,t) = \tilde{u}_x(t),\nonumber
\end{align}
and 
\begin{align}
\frac{\upartial^2 u(x,y,t)}{\upartial y^2} = 0 \hskip 5mm\mbox{yields}\hskip 5mm C_2(y,t)=\tilde{u}_y(t) y+\tilde{u}_0(t),\nonumber
\end{align}
and thus \eqref{uLinearCondition} is achieved. Similar arguments can be made to achieve (\ref{vLinearCondition},c), respectively. 
\end{proof}

\noindent We note the significance of Theorem \ref{thm:ADM_Ansatz_Connection}. In the works of  \cite{thacker1981some,shapiro1996nonlinear}, and \cite{matskevich2019exact} equations (\ref{uLinearCondition}-c)
were presented as \emph{ansatz solutions}, where they were also used to produce the  reduced system of shallow water equations to derive closed-form solutions. This theorem removes these assumptions and provides more insight to this behaviour by connecting it to the initial conditions (\ref{PartialDerivIC_x_Cond}-c).

\section{Novel exact solutions for flat bottom topographies with constant Coriolis force}
\label{sec:new_exact}

Next, we use the ADM construction to derive new families of solutions and their properties that describe other geophysical flows such as inertial oscillations and anticyclonic vortices which have a profound effect on oceanic and atmospheric dynamics \citep{vallis2017atmospheric, kafiabad2021interaction}. Here, we consider flows over flat bottom topologies where $D_0=0$ in \eqref{eq:idealBottomForm} with constant Coriolis parameter ($\bar{f} \ne 0$).

\subsection{Inertial oscillations in geostrophic flows}
\label{subsec:new_rotating}

For these types of flows, our analysis considers the following initial conditions.  
\begin{itemize}
	\item Condition I 
	\label{CondI}
	\begin{align}
		u_0(x,y)=v_0(x,y)=0,\hskip 5mm	h_0(x,y)= \eta_x x + \eta_y y,
	\label{eq:CondI_IC}
	\end{align}
	
	\item Condition II
	\label{CondII}
	\begin{align}
		u_0(x,y)=v_0(x,y)=0,\hskip 5mm h_0(x,y)= \eta_x x,
	\label{eq:CondII_IC}
	\end{align}
	
	\item Condition III
	\label{CondIII}
	\begin{align}
		u_0(x,y)=v_0(x,y)=0,\hskip 5mm h_0(x,y)= \eta_y y,
	\label{eq:CondIII_IC}
	\end{align}
\end{itemize}
where $\eta_x$ and $\eta_y$ are the respective constant free surface gradients in the $x$ and $y$ directions. We note that the behaviour of the initial conditions \eqref{eq:CondI_IC} - \eqref{eq:CondIII_IC} affect the decomposition of the decomposed functions of $u$, $v$, and $h$ as presented in the following lemma. 

\begin{lemma}
	Let $\{u_n(x,y,t)\}$, $\{v_n(x,y,t)\}$, $\{h_n(x,y,t)\}$ be the sequence of decomposed functions of $u$, $v$, and $h$ such that their relationship is defined by \eqref{eq:ADM_iter_Linv} (for $n \in \mathbb{N}$). If $D= 0$ and the initial conditions $u_0(x,y)$, $v_0(x,y)$, $h_0(x,y)$ satisfy the following properties
\begin{subequations}
    \begin{align}
        	\frac{\upartial u_0(x,y)}{\upartial x} =& \frac{\upartial v_0(x,y)}{\upartial x} = \frac{\upartial^2 h_0(x,y)}{\upartial x^2} = 0,
		\label{PartialDerivIC_x_Cond_sin}
	\\
 	\frac{\upartial u_0(x,y)}{\upartial y} =& \frac{\upartial v_0(x,y)}{\upartial y} = \frac{\upartial^2 h_0(x,y)}{\upartial y^2} = 0,
		\label{PartialDerivIC_y_Cond_sin}
	\\
 	\frac{\upartial^2 h_0(x,y)}{\upartial x \upartial y} =& 0.
		\label{PartialDerivIC_mixedxy_Cond_sin}
    \end{align}
\end{subequations}
Then the higher order components $u_n(x,y,t)$, $v_n(x,y,t)$, $h_n(x,y,t)$ also satisfy the property that  
\begin{subequations}
    \begin{align}
\frac{\upartial u_n(x,y,t)}{\upartial x} =& \frac{\upartial v_n(x,y,t)}{\upartial x} = \frac{\upartial h_n(x,y,t)}{\upartial x} = 0,\label{PartialDeriv_x_Cond_sin}
	\\
 \frac{\upartial u_n(x,y,t)}{\upartial y} =& \frac{\upartial v_n(x,y,t)}{\upartial y} = \frac{\upartial h_n(x,y,t)}{\upartial y} = 0\label{PartialDeriv_y_Cond_sin}
    \end{align}
\end{subequations}
for $n \in \mathbb{N}^{+}$.
		\label{ADM:PartialDerivativeCoeffLemma_sin}
\end{lemma}

\begin{proof}
	This is proven via mathematical induction by examining the recursion relationships for $u$, $v$, and $h$ in \eqref{eq:ADM_iter_Linv}. Condition \eqref{PartialDeriv_x_Cond_sin} is demonstrated by examining the following relationships
\begin{subequations}
    \begin{align}
        	\frac{\upartial u_{n+1}}{\upartial x}=&-{L}_t^{-1}\left\{\frac{\upartial}{\upartial x}\left[A_n\left(u,\frac{\upartial u}{\upartial x}\right) + A_n\left(v,\frac{\upartial u}{\upartial y}\right) \right] + \frac{1}{F^2}\frac{\upartial^{2} h_{n}}{\upartial x^{2}} - \bar{f}\, \frac{\upartial v_n}{\upartial x}\right\}, \hskip8mm
		\label{un_xx_sin}\\
  	\frac{\upartial v_{n+1}}{\upartial x} =&-{L}_t^{-1}\left\{\frac{\upartial}{\upartial x} \left[A_n\left(u,\frac{\upartial v}{\upartial x}\right) + A_n\left(v,\frac{\upartial v}{\upartial y}\right) \right] + \frac{1}{F^2}\frac{\upartial^{2} h_{n}}{\upartial x \upartial y} + \bar{f} \, \frac{\upartial u_n}{\upartial x} \right\},
		\label{vn_xx_sin}\\
  \frac{\upartial h_{n+1}}{\upartial x} =& - {L}_t^{-1}\left\{\frac{\upartial^{2} }{\upartial x^{2}}[A_n(u,h)] + \frac{\upartial^{2} }{\upartial x\upartial y}[A_n(v,h)]\right\}.
		\label{hn_xxx_sin}
    \end{align}
\end{subequations}
Therefore, when $n = 0$, equations (\ref{un_xx_sin}-c)
representing the relationship between the initial and first components for $u$, $v$, and $h$ become
\begin{subequations}
    \begin{align}
        \frac{\upartial u_{1}}{\upartial x}=&-{L}_t^{-1}\left\{\frac{\upartial}{\upartial x}\left[A_0\left(u,\frac{\upartial u}{\upartial x}\right) + A_0\left(v,\frac{\upartial u}{\upartial y}\right) \right] + \frac{1}{F^2}\frac{\upartial^{2} h_{0}}{\upartial x^{2}} - \bar{f}\, \frac{\upartial v_0}{\upartial x}\right\},\hskip8mm
		\label{u1_xx_sin}\\
  	\frac{\upartial v_{1}}{\upartial x} =&-{L}_t^{-1}\left\{\frac{\upartial}{\upartial x} \left[A_0\left(u,\frac{\upartial v}{\upartial x}\right) + A_0\left(v,\frac{\upartial v}{\upartial y}\right) \right] + \frac{1}{F^2}\frac{\upartial^{2} h_{0}}{\upartial x \upartial y} + \bar{f} \, \frac{\upartial u_0}{\upartial x} \right\},
		\label{v1_xx_sin}\\
  	\frac{\upartial h_{1}}{\upartial x} =& - {L}_t^{-1}\left\{\frac{\upartial^{2} }{\upartial x^{2}}[A_0(u,h)] + \frac{\upartial^{2} }{\upartial x\upartial y}[A_0(v,h)]\right\}.
		\label{h1_xxx_sin}
    \end{align}
\end{subequations}
Employing (\ref{PartialDerivIC_x_Cond_sin}-c)
it can be shown that equations (\ref{u1_xx_sin}-c)
reduce to the following relationship
	\begin{equation}
		\frac{\upartial u_1(x,y,t)}{\upartial x} = \frac{\upartial v_1(x,y,t)}{\upartial x} = \frac{\upartial h_1(x,y,t)}{\upartial x} = 0,\nonumber
	\end{equation}
and continuing this argument for $n \in \mathbb{N}^{+}$ yields equation  \eqref{PartialDeriv_x_Cond_sin}. Following similar arguments yields \eqref{PartialDeriv_y_Cond_sin}.
\end{proof}

\noindent From this, the behaviour of uniform $u$, $v$ over space, and planar free surface $h$ with constant spatial gradients over time can be summarised in the following theorem. 

\begin{theorem}
	\label{thm:ADM_order_n_sin}
	Let $\{u_n(x,y,t)\}$, $\{v_n(x,y,t)\}$, $\{h_n(x,y,t)\}$ be the sequence of decomposed functions of $u$, $v$, and $h$, where their relationship is defined by \eqref{eq:ADM_iter_Linv} (for $n \in \mathbb{N}$). If $D=0$ and the initial conditions $u_0(x,y)$, $v_0(x,y)$, $h_0(x,y)$ satisfy the properties defined in \rm{(\ref{PartialDerivIC_x_Cond_sin}-c)},
then the solutions $u$, $v$, and $h$ have the following properties 
\begin{subequations}
\label{ParticalCondition_all}
    \begin{align}
        \frac{\upartial u(x,y,t)}{\upartial x} = &\frac{\upartial u(x,y,t)}{\upartial y} = 0,
		\label{PartialCondition_u_sin}\\
  	\frac{\upartial v(x,y,t)}{\upartial x} =& \frac{\upartial v(x,y,t)}{\upartial y} 
		= 0,
		\label{PartialCondition_v_sin}\\
  	\frac{\upartial h(x,y,t)}{\upartial x}=& \frac{\upartial h(x,y,0)}{\upartial x},\\
   \frac{\upartial h(x,y,t)}{\upartial y}=& \frac{\upartial h(x,y,0)}{\upartial y},
		\label{ParticalCondition_h_sinA}\\
  	\frac{\upartial^2 h(x,y,t)}{\upartial x^2} =& \frac{\upartial^2 h(x,y,t)}{\upartial x\upartial y}=\frac{\upartial^2 h(x,y,t)}{\upartial y^2}=0. 
		\label{ParticalCondition_h_sinB}
    \end{align}
\end{subequations}
 Additionally, $u$, $v$, and $h$ are reduced to the following forms 
 \begin{subequations}
     \begin{align}
         {u}(x,y,t)=&\tilde{u}_0(t),
		\label{uLinearCondition_sin}\\
  	{v}(x,y,t)=&\tilde{v}_0(t),
		\label{vLinearCondition_sin}
	\\
 	{h}(x,y,t)=&\tilde{h}_0(t)+\tilde{h}_x x+\tilde{h}_yy,
		\label{hLinearCondition_sin}
     \end{align}
 \end{subequations}
 where the coefficients $\tilde{u}_0(t)$, $\tilde{v}_0(t)$, and  $\tilde{h}_0(t)$ are time-dependent, while $\tilde{h}_x$ and $\tilde{h}_y$ are constants. Additionally, (\ref{uLinearCondition_sin}-c)
 satisfy the reduced system of equations
	\begin{subequations}
	\label{eq:ODE_sin}
      \begin{align}
		\frac{\rm d}{{\rm d}t}\tilde{u}_0(t)=&-\frac{1}{F^2}\tilde{h}_x +\bar{f}\tilde{v}_0(t), \\
		\frac{\rm d}{{\rm d}t}\tilde{v}_0(t)=&-\frac{1}{F^2}\tilde{h}_y -\bar{f} \tilde{u}_0(t),\;\; \\
		\frac{\rm d}{{\rm d}t}\tilde{h}_0(t)=&-\tilde{u}_0(t) \tilde{h}_x -\tilde{v}_0(t)\tilde{h}_y. 
	\end{align}
 	\end{subequations}
\end{theorem}

\begin{proof}
  Applying Lemma \ref{ADM:PartialDerivativeCoeffLemma_sin} to each component in \eqref{eq:ADM_series} yields (\ref{ParticalCondition_all}).
  From \eqref{PartialCondition_u_sin}, we observe that 
	\begin{align}
		\frac{\upartial u(x,y,t)}{\upartial x} = 0 \hskip 5mm  \mbox{yields} \hskip 5mm u(x,y,t)=C_1(y,t),
	\nonumber
	\end{align}
where the integration constants, $C_1(y,t)$, are independent of $x$. Similarly, we have 
	\begin{align}
		\frac{\upartial u(x,y,t)}{ \upartial y} = 0 \hskip 5mm \mbox{yields} \hskip 5mm C_1(y,t) = \tilde{u}_0(t)
	\nonumber	
	\end{align}
        and thus \eqref{uLinearCondition_sin} is achieved. Similar arguments can be made to achieve (\ref{vLinearCondition_sin},c),
        respectively. Substituting (\ref{uLinearCondition_sin}-c)
        into \eqref{eq:SWE_non_dimen} achieves the reduced system of equations (\ref{eq:ODE_sin}), which completes the proof.
\end{proof}

\noindent Hence, we have the following results for inertial oscillations for geostrophic flows.

\begin{theorem}
	\label{thm:GeneralInertialOscillations}
	Given inertial oscillations over flat bottom topographies with constant Coriolis parameter $\bar{f} \ne 0$, where the initial behaviour is defined by \eqref{eq:CondI_IC}. The solutions $u$, $v$, and $h$ are expressed as 
\begin{subequations}
    \begin{align}
        u(x,y,t)=&-\frac{\eta_x}{\bar{f}F^2}\sin\left(\bar{f}t\right)-\frac{\eta_y}{\bar{f}F^2}\left[1-\cos\left(\bar{f}t \right)\right],
		\label{eq:exactCase12_u}\\
  	v(x,y,t)=&\frac{\eta_x}{\bar{f}F^2}\left[1-\cos\left(\bar{f}t\right)\right]-\frac{\eta_y}{\bar{f}F^2}\sin(\bar{f}t),
		\label{eq:exactCase12_v}\\
  	h(x,y,t)=&\frac{\eta_x^2}{\bar{f}^2F^2} \left[1-\cos \left(\bar{f}t \right) \right]+x\eta_x+\frac{\eta_y^2}{\bar{f}^2F^2} \left[1-\cos \left(\bar{f}t \right) \right]+\eta_y y,  
		\label{eq:exactCase12_h}
    \end{align}
\end{subequations}
 
	\noindent where $\eta_x$ and $\eta_y$ are the constant free surface gradients in the $x$ and $y$ directions, respectively.
\end{theorem}

\begin{proof}
  The initial conditions  \eqref{eq:CondI_IC} satisfy (\ref{PartialDerivIC_x_Cond_sin}-c).
  Therefore, the sequence of decomposed functions $\{u_n(x,y,t)\}$, $\{v_n(x,y,t)\}$, $\{h_n(x,y,t)\}$ satisfy (\ref{PartialDeriv_x_Cond_sin},b)
  for $n \in \mathbb{N}^{+}$ which satisfies Lemma \ref{ADM:PartialDerivativeCoeffLemma_sin} and consequently Theorem \ref{thm:ADM_order_n_sin}. Examining the system of reduced equations \eqref{eq:ODE_sin}, the initial conditions \eqref{eq:CondI_IC} also produce the following reduced relationships: $\tilde{h}_x=\eta_x$, $\tilde{h}_y=\eta_y$, and $\tilde{u}(t=0)=\tilde{v}(t=0)=\tilde{h}_0(t=0)=0$. Solving this reduced system achieves (\ref{eq:exactCase12_u}-c)
  which proves the theorem.
\end{proof}

\begin{corollary}
	\label{cor:OneSidedInertialOscillations}
	Given inertial oscillations over flat bottom topographies with constant Coriolis parameter $\bar{f} \ne 0$. 
	
	\begin{enumerate}[(i)]
		\item If the initial behaviour is defined by \eqref{eq:CondII_IC}, then the solutions $u$, $v$, and $h$ are expressed as 
\begin{subequations}
\label{all_Cor41i}
    \begin{align}
        	u(x,y,t)=&-\frac{\eta_x}{\bar{f}F^2}\sin\left(\bar{f}t\right),\\
			v(x,y,t)=&\frac{\eta_x}{\bar{f}F^2}\left[1-\cos\left(\bar{f}t\right)\right],
		\label{uVector_Cor41i}\\
  h(x,y,t)=&\frac{\eta_x^2}{\bar{f}^2F^2}\left[1-\cos\left(\bar{f}t\right)\right]+x\eta_x.
		\label{h_Cor41i}
    \end{align}
\end{subequations}
  
		\item If the initial behaviour is defined by \eqref{eq:CondIII_IC}, then the solutions $u$, $v$, and $h$ are expressed as  
\begin{subequations}
    \begin{align}
        u(x,y,t)=&-\frac{\eta_y}{\bar{f}F^2}\left[1-\cos\left(\bar{f}t \right)\right],\\
			v(x,y,t)=&-\frac{\eta_y}{\bar{f}F^2}\sin\left(\bar{f}t \right),\\
   h(x,y,t)=&\frac{\eta_y^2}{\bar{f}^2F^2}\left[1-\cos \left(\bar{f}t \right)\right]+\eta_y y.
    \end{align}
\end{subequations}
  
	\end{enumerate}

	\noindent $\eta_x$ and $\eta_y$ are the constant free surface gradients in the $x$ and $y$ directions, respectively.
\end{corollary}

\begin{proof}
	This is a special case of Theorem \ref{thm:GeneralInertialOscillations} for $\eta_y=0$ and $\eta_x=0$, respectively. 
\end{proof}

\par Theorem \ref{thm:GeneralInertialOscillations} and Corollary \ref{cor:OneSidedInertialOscillations} show the explicit relationship between these types of  flows with respect to the constant Coriolis parameter, the  free surface gradients, and the Froude number where the inertial oscillation frequency is defined by the constant Coriolis parameter $\bar{f}$. These results also demonstrate that these oscillations are based on the magnitude of the free surface gradients that depend on the initial behaviour and the geostrophic flows, which are consistent with the results of   \citep{vallis2017atmospheric}. Moreover, Theorem \ref{thm:GeneralInertialOscillations} describes these types of oscillations as the interaction between the geostrophic flow fluctuations and the free surface gradients, where Corollary \ref{cor:OneSidedInertialOscillations} considers cases when these gradients are negligible in the $x$ and $y$ directions. 

\subsection{Anticyclonic vortices with finite escape times}
\label{subsec:new_unbounded_shear}
For these types of flows our analysis considers the following initial conditions 
\begin{itemize}
	\item Condition IV
	\begin{align}
		u_0(x,y)=\bar{f}y, \hskip 5mm  v_0(x,y)=0,  \hskip 5mm h_0(x,y)= h_0,
		\label{eq:CondIV_IC}
	\end{align}
	\item Condition V
	\begin{align}
		u_0(x,y)=\bar{f}y,\hskip 5mm  v_0(x,y)=-\bar{f}x + \bar{f}y ,\hskip 5mm  h_0(x,y)= h_0,
		\label{eq:CondV_IC}
	\end{align}
	\item Condition VI
	\begin{align}
		u_0(x,y)=0, \hskip 5mm v_0(x,y)=-\bar{f}x,  \hskip 5mm h_0(x,y)= h_0,
		\label{eq:CondVI_IC}
	\end{align}
	\item Condition VII
	\begin{align}
		u_0(x,y)= \bar{f}x + \bar{f}y, \hskip 5mm v_0(x,y)=-\bar{f}x,  \hskip 5mm h_0(x,y)= h_0,
		\label{eq:CondVII_IC}
	\end{align}
\end{itemize}
where $h_0$ is the constant free surface height. These describe anticyclonic vortices for the initial vorticity is proportional to the negative constant Coriolis parameter. The behaviour of the initial conditions \eqref{eq:CondIV_IC} - \eqref{eq:CondVII_IC} affect the decomposition of the decomposed functions of $u$, $v$, and $h$ as presented in the following lemmas. 

\begin{lemma}
	Let $\{u_n(x,y,t)\}$, $\{v_n(x,y,t)\}$, $\{h_n(x,y,t)\}$ be the sequence of decomposed functions of $u$, $v$, and $h$, where their relationship is defined by \eqref{eq:ADM_iter_Linv} (for $n \in \mathbb{N}$) given a flat bottom topography $D=0$. If the initial conditions $u_0(x,y)$, $v_0(x,y)$, $h_0(x,y)$ are defined such that 
\begin{subequations}
    \begin{align}
        u_0(x,y)=&\bar{f}y,
\label{PartialDerivIC_x_Cond_tan_u_fy}\\
  \frac{\upartial^2 v_0(x,y)}{\upartial x^2}=&\frac{\upartial h_0(x,y)}{\upartial x}=0,	\label{PartialDerivIC_x_Cond_tan_u}\\
  	\frac{\upartial^2 v_0(x,y)}{\upartial y^2}=&\frac{\upartial h_0(x,y)}{\upartial y}=0,
\label{PartialDerivIC_y_Cond_tan_u}\\
  \frac{\upartial^2 v_0(x,y)}{\upartial x\upartial y}=&0.
\label{PartialDerivIC_mixedxy_Cond_tan_u}
    \end{align}
\end{subequations}
Then the higher order components $u_n(x,y,t)$, $v_n(x,y,t)$, $h_n(x,y,t)$, for $n \in \mathbb{N}^{+}$ satisfy
\begin{subequations}
    \begin{align}
        u_n(x,y,t)=&0,
		\label{PartialDeriv_x_Cond_tan_u_fy}\\
  	\frac{\upartial^2 v_n(x,y\rev{,t})}{\upartial x^2}=&\frac{\upartial h_n(x,y\rev{,t})}{\upartial x}=0,
		\label{PartialDeriv_x_Cond_tan_u}\\
  	\frac{\upartial^2 v_n(x,y\rev{,t})}{\upartial y^2}=& \frac{\upartial h_n(x,y\rev{,t})}{\upartial y}=0,
		\label{PartialDeriv_y_Cond_tan_u}\\
  \frac{\upartial^2 v_n(x,y\rev{,t})}{\upartial x\upartial y}=&0.
		\label{PartialDeriv_mixedxy_Cond_tan_u}
    \end{align}
\end{subequations}
 \label{ADM:PartialDerivativeCoeffLemma_tan_u}
\end{lemma}

\begin{proof}
	This is proven via mathematical induction by examining the recursion relationships for $u$, $v$, and $h$ in equation \eqref{eq:ADM_iter_Linv}. Condition \eqref{PartialDeriv_x_Cond_tan_u_fy} is demonstrated by examining
	\begin{equation}
		u_{n+1}=-{L}_t^{-1}\left\{\left[A_n\left(u,\frac{\upartial u}{\upartial x}\right) + A_n\left(v,\frac{\upartial u}{\upartial y}\right) \right] + \frac{1}{F^2}\frac{\upartial h_{n}}{\upartial x} - \bar{f}\, v_n\right\},
		\label{un_xx_tan_u}
	\end{equation}
In the case of $n=0$ and using \eqref{PartialDerivIC_x_Cond_tan_u_fy} - \eqref{PartialDerivIC_mixedxy_Cond_tan_u}, it reduces to 
	\begin{align}
		u_{1}&=-{L}_t^{-1}\left\{\left[A_0\left(u,\frac{\upartial u}{\upartial x}\right) + A_0\left(v,\frac{\upartial u}{\upartial y}\right) \right] + \frac{1}{F^2}\frac{\upartial h_{0}}{\upartial x} - \bar{f}\, v_0\right\}\nonumber\\
		&=-{L}_t^{-1}\left\{ A_0\left(v,\frac{\upartial u}{\upartial y}\right) - \bar{f}\, v_0\right\}\nonumber\\
		&=-{L}_t^{-1}\left\{v_0 \frac{\upartial u_0}{\upartial y} - \bar{f}\,v_0\right\}\nonumber \\
    & =0,\nonumber
	\end{align}
and continuing this argument for $n = \{1, 2, \ldots, n - 1\}$ yields equation  \eqref{PartialDeriv_x_Cond_tan_u_fy}. Condition \eqref{PartialDeriv_x_Cond_tan_u} is demonstrated by examining the following relationships
\begin{subequations}
    \begin{align}
    \frac{\upartial^2 v_{n+1}}{\upartial x^2} =&-{L}_t^{-1}\left\{\frac{\upartial^2}{\upartial x^2} \left[A_n\left(u,\frac{\upartial v}{\upartial x}\right) + A_n\left(v,\frac{\upartial v}{\upartial y}\right) \right] + \frac{1}{F^2}\frac{\upartial^{3} h_{n}}{\upartial x^2 \upartial y} + \bar{f} \, \frac{\upartial^2 u_n}{\upartial x^2 } \right\}, \hskip 10mm
		\label{vn_xx_tan_u}\\
        \frac{\upartial h_{n+1}}{\upartial x} =& - {L}_t^{-1}\left\{\frac{\upartial^{2} }{\upartial x^2}[A_n(u,h)] + \frac{\upartial^{2} }{\upartial x \upartial y}[A_n(v,h)]\right\}.
		\label{hn_xxx_tan_u}
    \end{align}
\end{subequations}
Therefore, when $n = 0$, equations (\ref{vn_xx_tan_u},b)
representing the relationship between the initial and first components for $v$ and $h$ become
 \begin{subequations}
     \begin{align}
         \frac{\upartial^2 v_{1}}{\upartial x^2} =& -{L}_t^{-1}\left\{\frac{\upartial^2}{\upartial x^2} \left[A_0\left(u,\frac{\upartial v}{\upartial x}\right) + A_0\left(v,\frac{\upartial v}{\upartial y}\right) \right] + \frac{1}{F^2}\frac{\upartial^{3} h_0}{\upartial x^2 \upartial y} + \bar{f} \, \frac{\upartial^2 u_0}{\upartial x^2} \right\},\hskip 5mm
		\label{v1_xx_tan_u}\\
  	\frac{\upartial h_{1}}{\upartial x} =& - {L}_t^{-1}\left\{\frac{\upartial^{2} }{\upartial x^2}[A_0(u,h)] + \frac{\upartial^{2} }{\upartial x \upartial y}[A_0(v,h)]\right\}.
		\label{h1_xxx_tan_u}
     \end{align}
 \end{subequations}
 Employing (\ref{PartialDerivIC_x_Cond_tan_u_fy}-d),
 it can be shown that equations
(\ref{v1_xx_tan_u},b)
 reduce to the following relationship
\begin{equation}
	\frac{\upartial^2 v_1(x,y,t)}{\upartial x^2} = \frac{\upartial h_1(x,y,t)}{\upartial x} = 0,\nonumber
\end{equation}
and continuing this argument for $n = \{1, 2, \ldots, n - 1\}$ yields equation  \eqref{PartialDeriv_x_Cond_tan_u}. Following similar arguments yields (\ref{PartialDeriv_y_Cond_tan_u},d).
\end{proof}

\begin{lemma}
	Let $\{u_n(x,y,t)\}$, $\{v_n(x,y,t)\}$, $\{h_n(x,y,t)\}$ be the sequence of decomposed functions of $u$, $v$, and $h$, where their relationship is defined by \eqref{eq:ADM_iter_Linv} (for $n \in \mathbb{N}$) given a flat bottom topography $D=0$. If the initial conditions $u_0(x,y)$, $v_0(x,y)$, $h_0(x,y)$ are defined such that 
\begin{subequations}
    \begin{align}
        v_0(x,y,t)=&-\bar{f}x ,
	\label{PartialDerivIC_x_tan_v_fx}
	\\
	\frac{\upartial^2 u_0(x,y)}{\upartial x^2}=&\frac{\upartial h_0(x,y)}{\upartial x}=0 ,
\label{PartialDerivIC_x_Cond_tan_v}\\
	\frac{\upartial^2 u_0(x,y)}{\upartial y^2}=&\frac{\upartial h_0(x,y)}{\upartial y}=0 ,
\label{PartialDerivIC_y_Cond_tan_v}\\
		\frac{\upartial^2 u_0(x,y)}{\upartial x\upartial y}=&0.
\label{PartialDerivIC_mixedxy_Cond_tan_v}
    \end{align}
\end{subequations}
 Then the higher order components $u_n(x,y,t)$, $v_n(x,y,t)$, $h_n(x,y,t)$, for $n  \in \mathbb{N}^{+}$ satisfy the property
	\begin{subequations}
	    \begin{align}
	    v_n(x,y,t)=&0,
		\label{PartialDeriv_x_Cond_tan_v_fx}\\  
	\frac{\upartial^2 u_n(x,y\rev{,t})}{\upartial x^2}=&\frac{\upartial h_n(x,y\rev{,t})}{\upartial x}=0,	\label{PartialDeriv_x_Cond_tan_v}\\
\frac{\upartial^2 u_n(x,y\rev{,t})}{\upartial y^2}=&\frac{\upartial h_n(x,y\rev{,t})}{\upartial y}=0,
\label{PartialDeriv_y_Cond_tan_v}\\
\frac{\upartial^2 u_n(x,y\rev{,t})}{\upartial x \upartial y}=&0.
\label{PartialDeriv_mixedxy_Cond_tan_v}
	\end{align}
	\end{subequations}
	
	\label{ADM:PartialDerivativeCoeffLemma_tan_v}
\end{lemma}

\begin{proof}
	This is proven via mathematical induction by examining the recursion relationships for $u$, $v$, and $h$ in equation \eqref{eq:ADM_iter_Linv}. Condition \eqref{PartialDeriv_x_Cond_tan_v_fx} is demonstrated by examining the following relationships
	\begin{equation}
		v_{n+1} =-{L}_t^{-1}\left\{ \left[A_n\left(u,\frac{\upartial v}{\upartial x}\right) + A_n\left(v,\frac{\upartial v}{\upartial y}\right) \right] + \frac{1}{F^2}\frac{\upartial h_{n}}{\upartial y} + \bar{f} \,  u_n \right\}.
		\label{un_xx_tan_v_fx}
	\end{equation}
At $n=0$, we have
	\begin{align}
		v_{1}& =-{L}_t^{-1}\left\{ \left[A_0\left(u,\frac{\upartial v}{\upartial x}\right) + A_0\left(v,\frac{\upartial v}{\upartial y}\right) \right] + \frac{1}{F^2}\frac{\upartial h_{0}}{\upartial y} + \bar{f} \,  u_0 \right\}\nonumber\\
		&=-{L}_t^{-1}\left\{ A_0\left(u,\frac{\upartial v}{\upartial x}\right) +\bar{f} \,  u_0 \right\}\nonumber\\
		&=-{L}_t^{-1}\left\{-u_0\bar{f} +\bar{f} \,  u_0 \right\}\nonumber\\
        &=0.\nonumber
	\end{align}
Employing a similar argument for $n=\{1,2,\ldots,n-1\}$, we have \eqref{PartialDeriv_x_Cond_tan_v_fx}. Equation \eqref{PartialDeriv_x_Cond_tan_v} is demonstrated by examining the following 
\begin{subequations}
    \begin{align}
        \frac{\upartial^{2} u_{n+1}}{\upartial x^{2}}=& -{L}_t^{-1}\left\{\frac{\upartial^{2}}{\upartial x^{2}}\left[A_n\left(u,\frac{\upartial u}{\upartial x}\right) + A_n\left(v,\frac{\upartial u}{\upartial y}\right) \right] + \frac{1}{F^2}\frac{\upartial^{3} h_{n}}{\upartial x^{3}} - \bar{f}\, \frac{\upartial^{2}v_n}{\upartial x^{2}}\right\},\hskip5mm
		\label{un_xx_tan_v}\\
  	\frac{\upartial h_{n+1}}{\upartial x} =& - {L}_t^{-1}\left\{\frac{\upartial^{2} }{\upartial x^{2}}[A_n(u,h)] + \frac{\upartial^{2} }{\upartial x \upartial y}[A_n(v,h)]\right\}.
		\label{hn_xxx_tan_v}
    \end{align}
\end{subequations}
Therefore, when $n = 0$, equations (\ref{un_xx_tan_v},b)
representing the relationship between the initial and first components for $u$ and $h$ become
\begin{subequations}
    \begin{align}
        \frac{\upartial^{2} u_{1}}{\upartial x^{2}}=& -{L}_t^{-1}\left\{\frac{\upartial^{2}}{\upartial x^{2}}\left[u_{0}\frac{\upartial u_{0}}{\upartial x} \,  + v_{0}\frac{\upartial u_{0}}{\upartial y} \right] + \frac{1}{F^2}\frac{\upartial^{3} h_{0}}{\upartial x^{3}} - \bar{f}\, \frac{\upartial^{2}v_{0}}{\upartial x^{2}}\right\},
		\label{u1_xx_tan_v}\\
  	\frac{\upartial h_{1}}{\upartial x} =& - {L}_t^{-1}\left\{\frac{\upartial^{2} }{\upartial x^{2}}[A_0(u,h)] + \frac{\upartial^{2} }{\upartial x \upartial y}[A_0(v,h)]\right\}.
		\label{h1_xxx_tan_v}
    \end{align}
\end{subequations}
Employing (\ref{PartialDerivIC_x_tan_v_fx}-d),
it can be shown that equations (\ref{u1_xx_tan_v},b)
reduce to the following relationship
	\begin{equation}
		\frac{\upartial^2 u_1(x,y,t)}{\upartial x^2} = \frac{\upartial h_1(x,y,t)}{\upartial x} = 0,\nonumber
		\label{PartialDerivComp1_x_Cond_tan_v}
	\end{equation}
        and continuing this argument for $n = \{1, 2, \ldots, n - 1\}$ yields equation  \eqref{PartialDeriv_x_Cond_tan_v}. Following similar arguments yields (\ref{PartialDeriv_y_Cond_tan_v},d).
\end{proof}

\noindent Therefore, the behaviour of $u$, $v$, and $h$ can be summarised in the following theorem. 

\begin{theorem}
	\label{thm:ADM_order_n_tan_u}
	Given a flat bottom topography, let $\{u_n(x,y,t)\}$, $\{v_n(x,y,t)\}$, $\{h_n(x,y,t)\}$ be the sequence of decomposed functions of $u$, $v$, and $h$, defined by \eqref{eq:ADM_iter_Linv}  (for $n \in \mathbb{N}$). If the initial conditions $u_0(x,y)$, $v_0(x,y)$, $h_0(x,y)$ are defined as {\rm{(\ref{PartialDerivIC_x_Cond_tan_u_fy}-d)}},
        then the solutions of $u$, $v$, and $h$ have the same property where 
\begin{subequations}
    \begin{align}        
		u(x,y,t)=&\bar{f}y,
\label{PartialCondition_u_tan_u_fy}\\
	\frac{\upartial^2 v(x,y,t)}{\upartial x^2}=&\frac{\upartial^2 v(x,y,t)}{\upartial y^2}=\frac{\upartial^2 v(x,y,t)}{\upartial x\upartial y}=0,
	\label{PartialCondition_v_tan_u}	\\
\frac{\upartial h(x,y,t)}{\upartial x}=&\frac{\upartial h(x,y,t)}{\upartial y}=0.\label{PartialCondition_h_tan_u}
    \end{align}
\end{subequations}
 Consequently, these solutions can be expressed as 
\begin{subequations}
    \begin{align}
 		{u}(x,y,t)=&\bar{f}y,		\label{uLinearCondition_tan_u}\\
		{v}(x,y,t)=&\tilde{v}_0(t)+\tilde{v}_x(t)x+\tilde{v}_y(t)y,\label{vLinearCondition_tan_u}	\\
		{h}(x,y,t)=&\tilde{h}_0(t), \label{hLinearCondition_tan_u}
    \end{align}
\end{subequations}
where the coefficients  $\tilde{v}_0(t)$, $\tilde{v}_x(t)$, $\tilde{v}_y(t)$, and $\tilde{h}_0(t)$ are time-dependent that also satisfy the following reduced system of equations
	\begin{subequations}
	\label{eq:ODE_tan_u}
 \begin{align}
		\frac{{\rm d}}{{\rm d}t}\tilde{v}_0(t)=&-\tilde{v}_0(t) \tilde{v}_y(t),  \\
		\frac{{\rm d}}{{\rm d}t}\tilde{v}_x(t)=&-\tilde{v}_x(t)\tilde{v}_y(t),  \\
		\frac{{\rm d}}{{\rm d}t}\tilde{v}_y(t)=&-\bar{f}\tilde{v}_x(t)-\tilde{v}_y(t)^2-\bar{f}^2 , \\
		\frac{{\rm d}}{{\rm d}t}\tilde{h}_0(t)=&-\tilde{h}_0(t)\tilde{v}_y(t).
	\end{align}
 	\end{subequations}
\end{theorem}

\begin{proof}
  Applying Lemma \ref{ADM:PartialDerivativeCoeffLemma_tan_u} to each component in \eqref{eq:ADM_series} yields
(\ref{PartialCondition_u_tan_u_fy}-c).
  From \eqref{PartialCondition_v_tan_u}, we observe that 
	\begin{align}
		\frac{\upartial^2 v(x,y,t)}{\upartial x^2} = 0\hskip 5mm \mbox{yields} \hskip 5mm v(x,y,t)=C_1(y,t) x+C_2(y,t),\nonumber
	\end{align}
        where the integration constants, $C_1(y,t)$ and $C_2(y,t)$, are independent of $x$. Similarly, we have 
	\begin{align}
		\frac{\upartial^2 v(x,y,t)}{\upartial x \upartial y} = 0 \hskip 5mm \mbox{yields} \hskip 5mm C_1(y,t) = \tilde{v}_x(t)\nonumber
	\end{align}
	 and 
	\begin{align}
		\frac{\upartial^2 v(x,y,t)}{\upartial y^2} = 0 \hskip 5mm \mbox{yields} \hskip 5mm C_2(y,t)=\tilde{v}_y(t) y+\tilde{v}_0(t).\nonumber
	\end{align}
        and thus \eqref{vLinearCondition_tan_u} is achieved. Similar arguments can be made to achieve
(\ref{uLinearCondition_tan_u},c),
        respectively. The reduced system of equations  \eqref{eq:ODE_tan_u} is obtained via substituting
(\ref{uLinearCondition_tan_u}-c)
        into \eqref{eq:SWE_non_dimen}.
\end{proof}

\begin{theorem}
	\label{thm:ADM_order_n_tan_v}
	Let $\{u_n(x,y,t)\}$, $\{v_n(x,y,t)\}$, $\{h_n(x,y,t)\}$ be the sequence of decomposed functions of $u$, $v$, and $h$, where their relationship is defined by \eqref{eq:ADM_iter_Linv} (for $n \in \mathbb{N}$) given a flat bottom topography $D=0$. If the initial conditions $u_0(x,y)$, $v_0(x,y)$, $h_0(x,y)$ are defined as
{\rm{(\ref{PartialDerivIC_x_tan_v_fx}-d)}},
then the solutions of $u$, $v$, and $h$ have the same property, where 
\begin{subequations}
    \begin{align}        
	\frac{\upartial^2 u(x,y,t)}{\upartial x^2} = &\frac{\upartial^2 u(x,y,t)}{\upartial y^2} 
		= \frac{\upartial^2 u(x,y,t)}{\upartial x \upartial y} = 0,
	\label{PartialCondition_u_tan_v}\\
			v(x,y,t)=& -\bar{f}x,
	\label{PartialCondition_v_tan_v}\\
		\frac{\upartial h(x,y,t)}{\upartial x}
		=& \frac{\upartial h(x,y,t)}{\upartial y} = 0.
	\label{PartialCondition_h_tan_v}
    \end{align}
\end{subequations}
 Consequently, these solutions can be expressed as 
\begin{subequations}
    \begin{align}
		{u}(x,y,t)=&\tilde{u}_0(t)+\tilde{u}_x(t)x+\tilde{u}_y(t)y,
	\label{uLinearCondition_tan_v}\\
		{v}(x,y,t)=&-\bar{f}x,
	\label{vLinearCondition_tan_v}\\
		{h}(x,y,t)=&\tilde{h}_0(t), 
	\label{hLinearCondition_tan_v}
    \end{align}
\end{subequations}
 where the coefficients $\tilde{u}_0(t)$, $\tilde{u}_x(t)$, $\tilde{u}_y(t)$, and $\tilde{h}_0(t)$,  are time-dependent. These coefficients satisfy
 \begin{subequations}
 	\label{eq:ODE_tan_v}
     \begin{align}
		\frac{\rm d}{{\rm d}t}\tilde{u}_0(t)=&-\tilde{u}_0(t) \tilde{u}_x(t), \\
		\frac{\rm d}{{\rm d}t}\tilde{u}_x(t)=&-\tilde{u}_x(t)^2+\bar{f}\tilde{u}_y(t)-\bar{f}^2,\\
		\frac{\rm d}{{\rm d}t}\tilde{u}_y(t)=&- \tilde{u}_y(t) \tilde{u}_x(t) ,\\
		\frac{\rm d}{{\rm d}t}\tilde{h}_0(t)=&-\tilde{h}_0(t)\tilde{u}_x(t).
	\end{align}
  \end{subequations}
\end{theorem}

\begin{proof}
  Applying Lemma \ref{ADM:PartialDerivativeCoeffLemma_tan_v} to each component in \eqref{eq:ADM_series} yields
(\ref{PartialCondition_u_tan_v}-c).
  From \eqref{PartialCondition_u_tan_v}, we observe that 
	\begin{align}
		\frac{\upartial^2 u(x,y,t)}{\upartial x^2} = 0 \hskip 5mm \mbox{yields}\hskip 5mm  u(x,y,t)=C_1(y,t) x+C_2(y,t),\nonumber
	\end{align}
where the integration constants, $C_1(y,t)$ and $C_2(y,t)$, are independent of $x$. Similarly, we have 
	\begin{align}
		\frac{\upartial^2 u(x,y,t)}{\upartial x \upartial y} = 0 \hskip 5mm \mbox{yields} \hskip 5mm C_1(y,t) = \tilde{u}_x(t)\nonumber
	\end{align}
and 
	\begin{align}
		\frac{\upartial^2 u(x,y,t)}{\upartial y^2} = 0 \hskip 5mm \mbox{yields} \hskip 5mm  C_2(y,t)=\tilde{u}_y(t) y+\tilde{u}_0(t).\nonumber
	\end{align}
        and thus \eqref{uLinearCondition_tan_v} is achieved. Similar arguments can be made to achieve
(\ref{vLinearCondition_tan_v},c).
        The reduced equations \eqref{eq:ODE_tan_v} is obtained by substituting
 (\ref{uLinearCondition_tan_v}-c)
        into \eqref{eq:SWE_non_dimen}. 
\end{proof}

\noindent Therefore, the following results describe closed-form solutions for anticyclonic vortices with finite escape times. 

\begin{theorem}
\noindent For any flows over flat bottom topographies ($D=0$) with a constant Coriolis parameter ($\bar{f} \ne 0$) and initial constant free surface height ($h_0$), the solutions $u$, $v$, and $h$ with respect to their corresponding initial conditions are defined as follows.
		\begin{enumerate}[(i)]
		\item If the initial behaviour is defined by \eqref{eq:CondIV_IC} then 
  \begin{subequations}
		\label{eq:CondIV_exact}
  \begin{align}
			u(x,y,t)=&\bar{f}y,\\
			v(x,y,t)=&-\bar{f}y\,\tan(\bar{f}t),\\
			h(x,y,t)=& h_0\,\sec(\bar{f}t).
		\end{align}
  \end{subequations}
\item If the initial behaviour is defined by \eqref{eq:CondV_IC} then \begin{subequations}
  		\label{eq:CondV_exact}
	\begin{align}
			u(x,y,t)=&\bar{f}\,y,\\
			v(x,y,t)=&\bar{f}y\sec(\bar{f}t)-\bar{f}y\tan(\bar{f}t)+x\left[-\frac{\rm d}{{\rm d}t}\tan(\bar{f}t)+\frac{\rm d}{{\rm d}t}\sec(\bar{f}t)\right], \\
			h(x,y,t)=&\frac{h_0}{\bar{f}}\left[\frac{\rm d}{{\rm d}t}\tan(\bar{f}t)-\frac{\rm d}{{\rm d}t}\sec(\bar{f}t)\right]. 
		\end{align}
  \end{subequations}
 \end{enumerate}
Furthermore, these solutions describe anticyclonic vortices with finite escape times that are based on the initial zonal velocity being represented as $u(x, y, 0) = u_0(x,y) = \bar{f} y$. 
\label{Thm:AnticyclonicVorA}
\end{theorem}

\begin{proof}
\noindent Equations \eqref{eq:CondIV_IC} and \eqref{eq:CondV_IC} satisfy Theorem \ref{thm:ADM_order_n_tan_u}, where these flows can be represented by  \eqref{eq:ODE_tan_u}. The initial conditions \eqref{eq:CondIV_IC} require 
\begin{subequations}
\label{ReducedICs_CondIV_all}
    \begin{align}
\tilde{v}_0(t=0)=&\tilde{v}_x(t=0)=\tilde{v}_y(t=0)=0,
\label{ReducedICs_CondIVa}\\
\tilde{h}_0(t=0)=&h_0. 
\label{ReducedICs_CondIVb}
    \end{align}
\end{subequations}
 Similarly, the initial conditions \eqref{eq:CondV_IC} require 
\begin{subequations}
\label{ReducedICs_CondV_all}
    \begin{align}
        \tilde{v}_0(t=0)=&0, \hskip 5mm  \tilde{v}_x(t=0)=-\bar{f}, \hskip 5mm   \tilde{v}_y(t=0)=\bar{f},
\label{ReducedICs_CondVa}\\
\tilde{h}_0(t=0)=&h_0. 
\label{ReducedICs_CondVb}
    \end{align}
\end{subequations}
Solving \eqref{eq:ODE_tan_u} with the initial conditions, defined by
\eqref{ReducedICs_CondIV_all} and \eqref{ReducedICs_CondV_all}, achieves \eqref{eq:CondIV_exact} and \eqref{eq:CondV_exact}, respectively. 
\end{proof}

\begin{theorem}
\noindent For any flows over flat bottom topographies ($D=0$) with a constant Coriolis parameter ($\bar{f} \ne 0$) and initial constant free surface height ($h_0 \ne 0$), the solutions $u$, $v$, and $h$ with respect to their corresponding initial conditions are defined as follows.

\begin{enumerate}[(i)]
	\item If the initial behaviour is defined by \eqref{eq:CondVI_IC} then 
	\begin{subequations}
	\label{eq:CondVI_exact}
   \begin{align}
		u(x,y,t)=&-\bar{f}x\,\tan(\bar{f}t),\\	v(x,y,t)=&-\bar{f}x,\\ h(x,y,t)=&h_0\,\sec(\bar{f}t).
	\end{align}  
	\end{subequations}
	\item If the initial behaviour is defined by \eqref{eq:CondVII_IC} then 
	\begin{subequations}
	\label{eq:CondVII_exact}	    
	\begin{align}
	u(x,y,t)=&\bar{f}x\sec(\bar{f}t)-\bar{f}x\tan(\bar{f}t)+y\left[\frac{{\rm d}}{{\rm d}t}\tan(\bar{f}t)-\frac{{\rm d}}{{\rm d}t}\sec(\bar{f}t)\right],\\
		v(x,y,t)=&-\bar{f} x,\\  h(x,y,t)=&\frac{h_0}{\bar{f}}\left[\frac{{\rm d}}{{\rm d}t}\tan(\bar{f}t)-\frac{{\rm d}}{{\rm d}t}\sec(\bar{f}t)\right].
	\end{align}
 	\end{subequations}
\end{enumerate}
Furthermore, these solutions describe anticyclonic vortices with finite escape times that are based on the initial meridional velocity being represented as $v(x,y, 0)  = v_0(x,y) = -\bar{f} x$. 
\label{Thm:AnticyclonicVorB}
\end{theorem}

\begin{proof}
\noindent Equations \eqref{eq:CondVI_IC} and \eqref{eq:CondVII_IC} satisfy Theorem \ref{thm:ADM_order_n_tan_v}, where these flows can be represented by \eqref{eq:ODE_tan_v}. The initial conditions \eqref{eq:CondVI_IC} require 
\begin{subequations}
\label{ReducedICs_CondVI_all}
    \begin{align}
\tilde{u}_0(t=0)=&\tilde{u}_x(t=0)=\tilde{u}_y(t=0)=0,\;\;
\label{ReducedICs_CondVIa}\\
\tilde{h}_0(t=0)=&h_0. 
\label{ReducedICs_CondVIb}
    \end{align}
\end{subequations}
\noindent Similarly, the initial conditions \eqref{eq:CondVII_IC} require 
\begin{subequations}
\label{ReducedICs_CondVII_all}
    \begin{align}
        \tilde{u}_0(t=0)=&0, \; \; \tilde{u}_x(t=0)=\tilde{u}_y(t=0)=\bar{f}, \;\;
\label{ReducedICs_CondVIIa}\\
\tilde{h}_0(t=0)=&h_0. 
\label{ReducedICs_CondVIIb}
    \end{align}
\end{subequations}
\noindent Solving \eqref{eq:ODE_tan_v} with the initial conditions, defined by
(\ref{ReducedICs_CondVI_all}) and (\ref{ReducedICs_CondVII_all}),
achieves \eqref{eq:CondVI_exact} and \eqref{eq:CondVII_exact}, respectively.
\end{proof}

\par Theorems \ref{Thm:AnticyclonicVorA} and \ref{Thm:AnticyclonicVorB} show that the flow velocity components directly depend only on the constant Coriolis parameter whereas the free surface height depends on both the constant Coriolis parameter and the initial free surface height. Since these solutions are valid for $t \in \left[0,  \pi/\left(2 \bar{f}\right) \right)$, these results also represent anticyclonic vortices with finite escape times that rotate faster and are more unstable than cyclonic ones which is consistent with previous observations \citep{tsang2015ellipsoidal, mckiver2020balanced}.  These solutions also consider the nonlinear balance between the inertial and Coriolis terms in the momentum portion of the shallow water equations,  which is important to understand irregularities between cyclonic and anticyclonic vortices which also improves previous results using quasi-geostrophic approximations \citep{vallis2019essentials, mckiver2020balanced},  linear stability analysis techniques \citep{clark2013improved}, and numerical approaches \citep{tsang2015ellipsoidal}.

\section{Numerical validation and results}
\label{sec:numerical}

Numerical validation is provided via examining the convergence and accuracy of the partial sums of $u$, $v$, and $h$ (given by $S_N(u)$, $S_N(v)$, and $S_N(h)$) against the governing equations \eqref{eq:SWE_non_dimen}, the exact solutions ($u$, $v$, and $h$), and numerical solutions ($\hat{u}$, $\hat{v}$, and $\hat{h}$) via the relative integral squared error defined as 
\begin{align}
	E(N) = \frac{\int_{-L_x}^{L_x}\int_{-L_y}^{L_y}\int_0^{T} e(N;x,y,t) \,{\rm d}t\,{\rm d}x\,{\rm d}y}{\int_{-L_x}^{L_x}\int_{-L_y}^{L_y}\int_0^{T}(u^2+v^2+h^2) \,{\rm d}t\,{\rm d}x\,{\rm d}y},
\label{eq:square_residue_int}
\end{align}

\noindent where $L_x=1$, $L_y=1$, and $T=1$. The convergence $E_c(N)$ is measured by evaluating \eqref{eq:square_residue_int} with 
\begin{align}
	& \hskip -10mm e(N;x,y,t) \nonumber\\
 =&\Bigg\{\frac{\upartial S_N(u)}{\upartial t}+S_N(u)\frac{\upartial S_N(u)}{\upartial x}+S_N(v)\frac{\upartial S_N(u)}{\upartial y}+\frac{1}{F^2}\frac{\upartial S_N(h)}{\upartial x}-\bar{f}S_N(v)\Bigg\}^2\nonumber\\
	&+\Bigg\{\frac{\upartial S_N(v)}{\upartial t}+S_N(u)\frac{\upartial S_N(v)}{\upartial x}+S_N(v)\frac{\upartial S_N(v)}{\upartial y}+\frac{1}{F^2}\frac{\upartial S_N(h)}{\upartial y}+\bar{f}S_N(u)\Bigg\}^2\hskip 15mm\nonumber\\
	&+\Bigg\{\frac{\upartial S_N(h)}{\upartial t}+\frac{\upartial }{\upartial x}[S_N(u)(S_N(h)+D)]+\frac{\upartial }{\upartial y}[S_N(v)(S_N(h)+D)]\Bigg\}^2.
\label{eq:SquareResidue_Conv}
\end{align}
$E_{ex}(N)$ is the accuracy of the partial sums of  $u$, $v$, and $h$ against the exact solutions which is measured via evaluating \eqref{eq:square_residue_int} with 
\begin{equation}
	e(N;x,y,t) = \left( S_N(u)- u \right)^2+ \left( S_N(v)- v\right)^2+ \bigl( S_N(h)- h\bigr)^2.
	\label{ADMExact_Residue}
\end{equation}
$\hat{E}(N)$ is the accuracy of the partial sums of  $u$, $v$, and $h$ against the numerical solutions which is measured via evaluating \eqref{eq:square_residue_int} with 
\begin{equation}
	e(N;x,y,t) = \left( S_N(u)- \hat{u} \right)^2+ \left( S_N(v)- \hat{v}\right)^2+ \bigl( S_N(h)- \hat{h}\bigr)^2.
\label{ADM_Residue}
\end{equation}
$\hat{E}_{ex}$  is the accuracy between the numerical and exact solutions, which is measured via evaluating \eqref{eq:square_residue_int} with 
\begin{equation}
	e(N;x,y,t) =	\left( u-\hat{u}\right)^2+ \left( v- \hat{v} \right)^2+ \bigl( h- \hat{h} \bigr)^2.
\label{TrueSoln_Residue}
\end{equation}

\par In all evaluations, we follow \citet{matskevich2019exact} where $F=1$ represents the characteristic velocity as $U_0=\sqrt{gH_0}$.  The summaries of all parameters used for our evaluations are listed in Table \ref{tab:my_label} below. Equation \eqref{eq:square_residue_int} is discretised with spatial grid spacings of $\Delta x=0.1$ and $\Delta y=0.1$ and a temporal grid spacing of $\Delta t=0.1$.  Numerical implementations ($\hat{u}$, $\hat{v}$, and $\hat{h}$) are done using the large-particle method as outlined by \cite{matskevich2019exact}. 

\begin{table}[!h]
	\centering
	\caption[center]{Summary of evaluation parameters, initial conditions, and applicable exact solutions used to validate Conditions I-VII.}
	\vspace{2mm}
	\begin{tabular}{cccccccc}
		\hline 
		Condition &$F$&  $\bar{f}$ & $D_0$ & $L$ & $l$ &  Other Parameters & Exact solutions \\
		\hline
		I  & 1 & 0.5 & 0 & - & -  & $\eta_x=10^{-4}$ & Theorem 3.3\\
		II  & 1 & 0.5  & 0& - & -  & $\eta_y=10^{-4}$ & Corollary 3.4(i)\\
		III & 1 & 0.5  & 0& - & -  &  $\eta_x=\eta_y=10^{-4}$ & Corollary 3.4(ii)\\
		IV  & 1 & 0.5  & 0& - & -  &  $h_0=10^{-4}$ & Theorem 3.9(i)\\
		V & 1 & 0.5  & 0& - & -  &  $h_0=10^{-4}$ & Theorem 3.9(ii)\\
		VI  & 1 & 0.5  & 0& - & -  &  $h_0=10^{-4}$ & Theorem 3.10(i)\\
		VII & 1 & 0.5  & 0& - & -  & 
		$h_0=10^{-4}$ & Theorem 3.10(ii)\\
		\hline
	\end{tabular}
\label{tab:my_label}
\end{table}

\subsection{Results}

\par Table \ref{tab:error_all} presents a summary of the convegence and accuracy results, where the partial sums (for $N = 2, 4$ and $6$) was used to assess the level of convergence. We note the convergence trend where the relative error margins stabilise between ${\rm O}\left(10^{-11}\right)$ and ${\rm O}\left(10^{-6}\right)$ at $N = 6$, which indicate that the Adomian approximations of up to six terms in its partial sum yield effective and robust estimates for Conditions I-VII. This is further validated when examining the accuracy of these partial sums with the numerical solutions, where the accuracies range between ${\rm O}\left(10^{-6}\right)$ and ${\rm O}\left(10^{-4}\right)$. We also note the comparisons between the explicit solutions generated for Conditions I-VII and the numerical solutions, where these deviations are also miniscule. 

\begin{table}[h]
    \caption{Summary of convergence trend (for $N = 2, 4$ and $6$) and accuracy (for $N = 6$) via integral squared error $E(N)$ calculations for Conditions I-VII.}
	\vspace{2mm}
	\centering
	\begin{tabular}{cllllll}
		\hline
		& $E_c(N=2)$ & $E_c(N=4)$ &   $E_c(N=6)$ & $E_{ex}(N = 6)$ & $\hat{E}(N = 6)$ & $\hat{E}_{ex}$ \\
		\hline
		Minimum & $3.3\times 10^{-3}$ & $1.5\times 10^{-6}$ & $8.2\times 10^{-11}$ &$1.6\times 10^{-12}$ & $3.1\times 10^{-6}$ & $3.1\times 10^{-6}$ \\
		Maximum  & $6.5\times 10^{-3}$ & $2.5\times 10^{-4}$ & $7.0\times 10^{-6}$ & $1.3\times 10^{-7}$ & $2.1\times 10^{-4}$ & $2.1\times 10^{-4}$\\
		\hline
	\end{tabular}
\label{tab:error_all}
\end{table} 

\par Figures  \ref{fig:error_case_12} through \ref{fig:error_case_15} present the behaviour of the ADM partial sums of $u$, $v$, and $h$  (for $N = 6$) along with Conditions II and IV (section \ref{sec:new_exact}) are used as examples. In each case we note the direct relationship between the initial conditions (figures \ref{fig:error_case_12}-\ref{fig:error_case_15} part (a)) and a temporal snapshot of the behaviour of the corresponding partial sums at $t = 1$ (figures \ref{fig:error_case_12}-\ref{fig:error_case_15} part (b)), which illustrates the velocity vector field $\boldsymbol{u} = \left<u\left(x,y,1 \right),v\left(x,y,1\right)\right>$ over the contour representing the free surface height $h\left(x,y,1\right)$. Figure \ref{fig:error_case_12}(a) shows the initial zero velocity over constant free surface gradient $\eta_x=10^{-4}$, which corresponds to the initial conditions represented by \eqref{eq:CondII_IC}. Figure \ref{fig:error_case_12}(b) confirms the temporal behaviour where we note the behaviour of $\boldsymbol{u}$ over the contour, which is analogous to the exact solutions described in equation \eqref{all_Cor41i}.
However, in figure \ref{fig:error_case_12}(b) we also observe the rotating velocity field $\boldsymbol{u}$ over the contour illustrating the behaviour of inertial geostrophic oscillations. These effects are not only driven by the pressure gradient due to variations in the free surface height but also due to the Coriolis force, which are also noticed analytically when constructing the ADM decompositions. These confirmations continue in figure \ref{fig:error_case_15}, where part (a) illustrates the behaviour of the initial velocity $\boldsymbol{u}_0 = \left<u\left(x,y,0 \right),v\left(x,y,0\right)\right>$ with respect to the initial free surface height $h_0 = h\left(x,y,0\right)$. Figure \ref{fig:error_case_15} also shows the correlation between the initial conditions and analytical solutions to Condition IV while also illustrating the effects of anticyclonic vortices with finite escape time as shown in figure \ref{fig:error_case_15} (b).  Specifically, we note the clockwise orientation of  $\boldsymbol{u}$ that is consistent with the behaviour of anticyclonic vortices which are valid for $t \in \left[0,  \pi/\left(2 \bar{f}\right) \right)$.

\begin{figure}[h]
	
	(a) \hspace{3in} (b)
	
	\centering
	\includegraphics[width=2.9in]{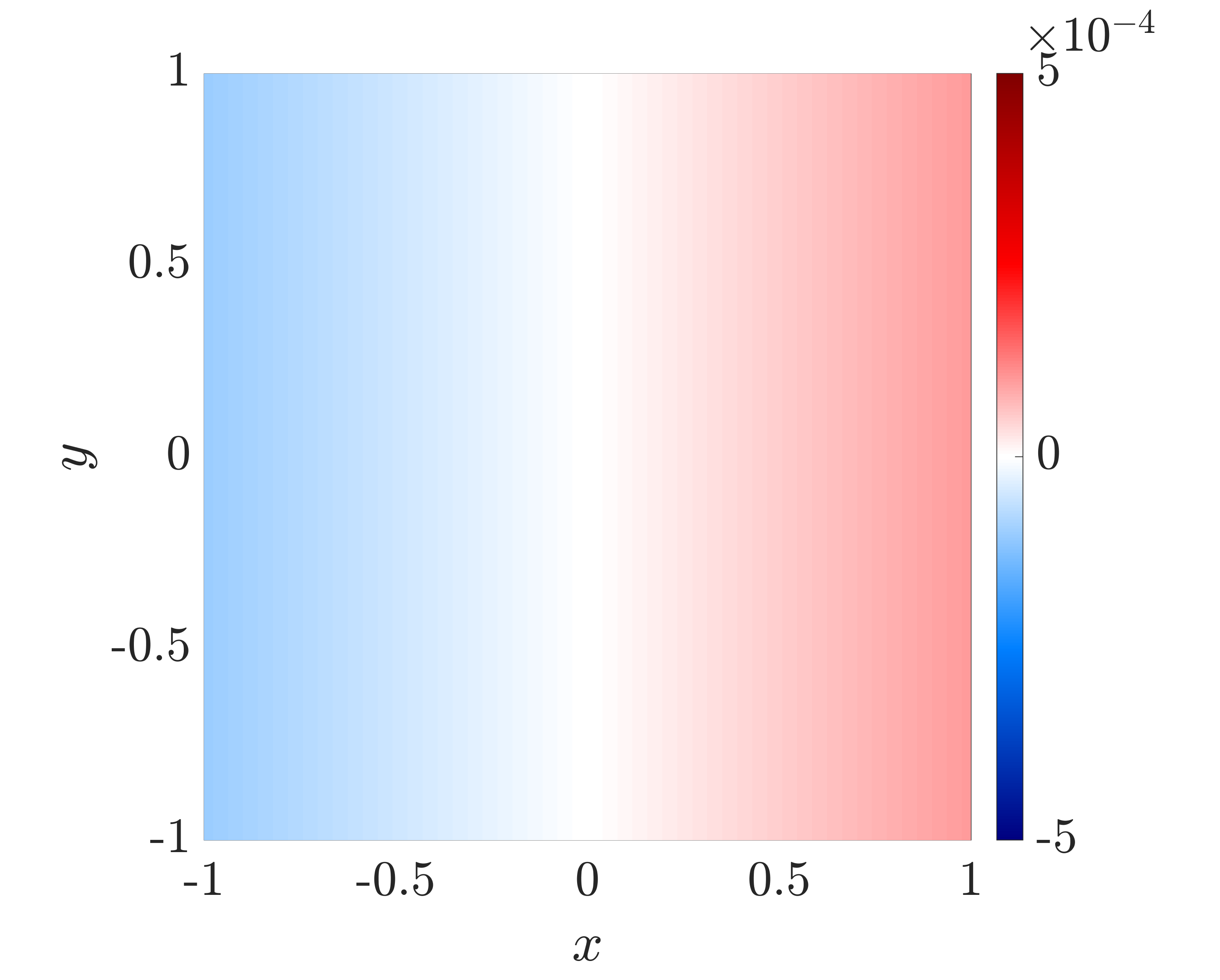}
	\includegraphics[width=2.9in]{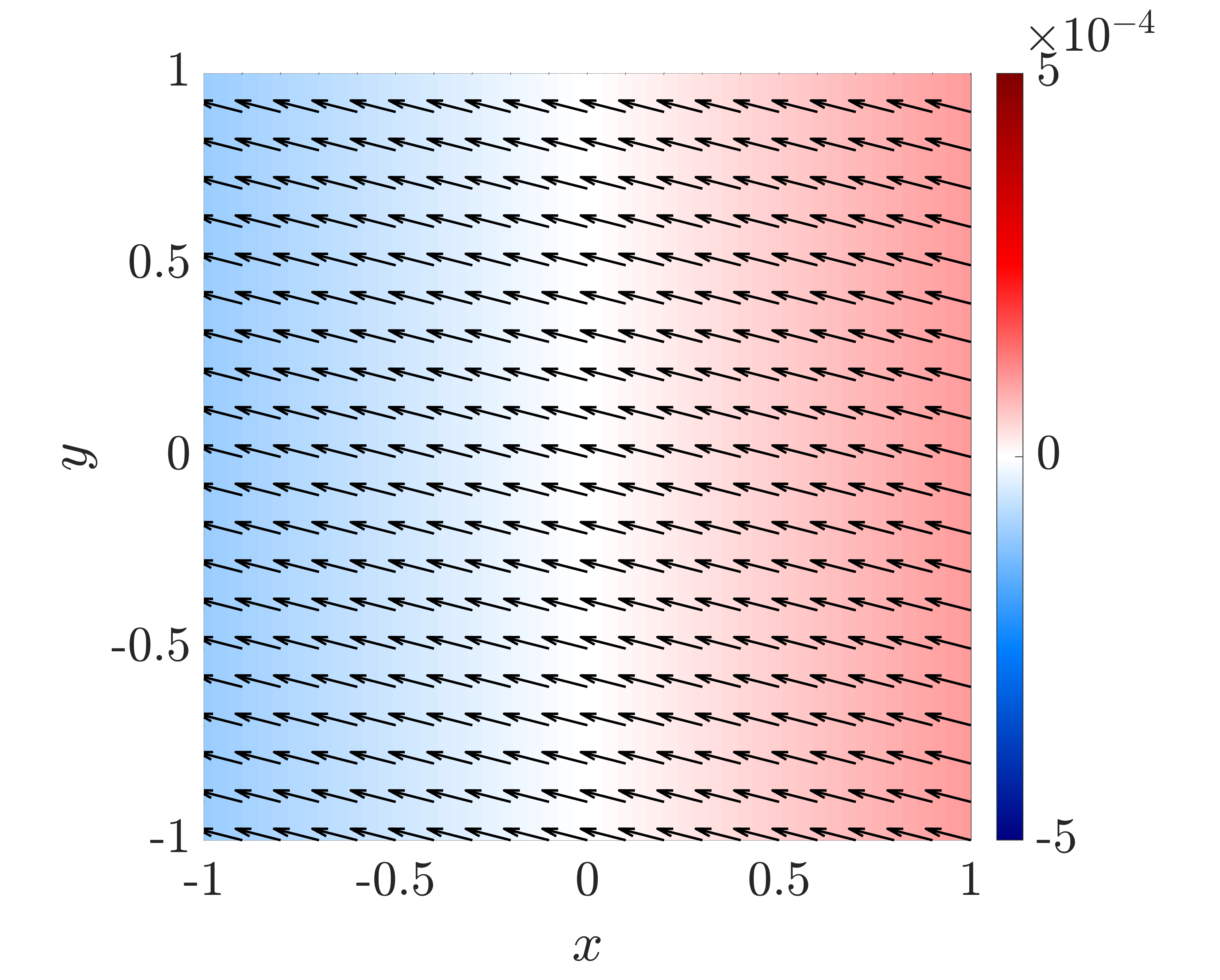}
	\caption{Velocity vector field $\boldsymbol{u}$ (arrow) and free surface height $h$ (contour) behaviour for Condition II: (a) initial condition at $t=0$ and (b) partial sum approximation based on ADM (with $N=6$) at $t=1$. Parameters used include $F=1$, $\bar{f}=0.5$, $D_0=0$, and $\eta_x=10^{-4}$ (Colour online).}
\label{fig:error_case_12}
\end{figure}

\begin{figure}[h]
	
	(a) \hspace{3in} (b)
	
	\centering
	\includegraphics[width=2.9in]{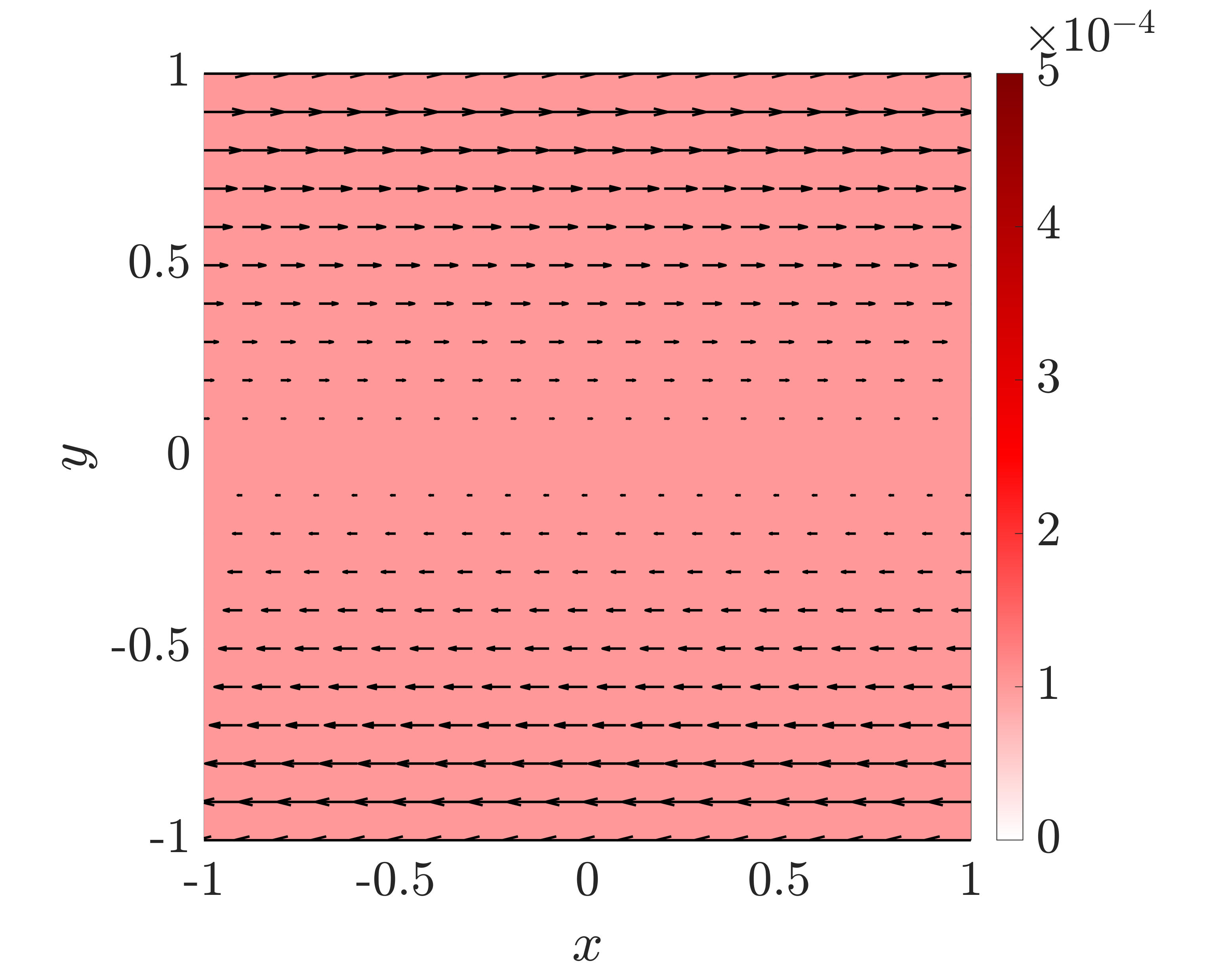}
	\includegraphics[width=2.9in]{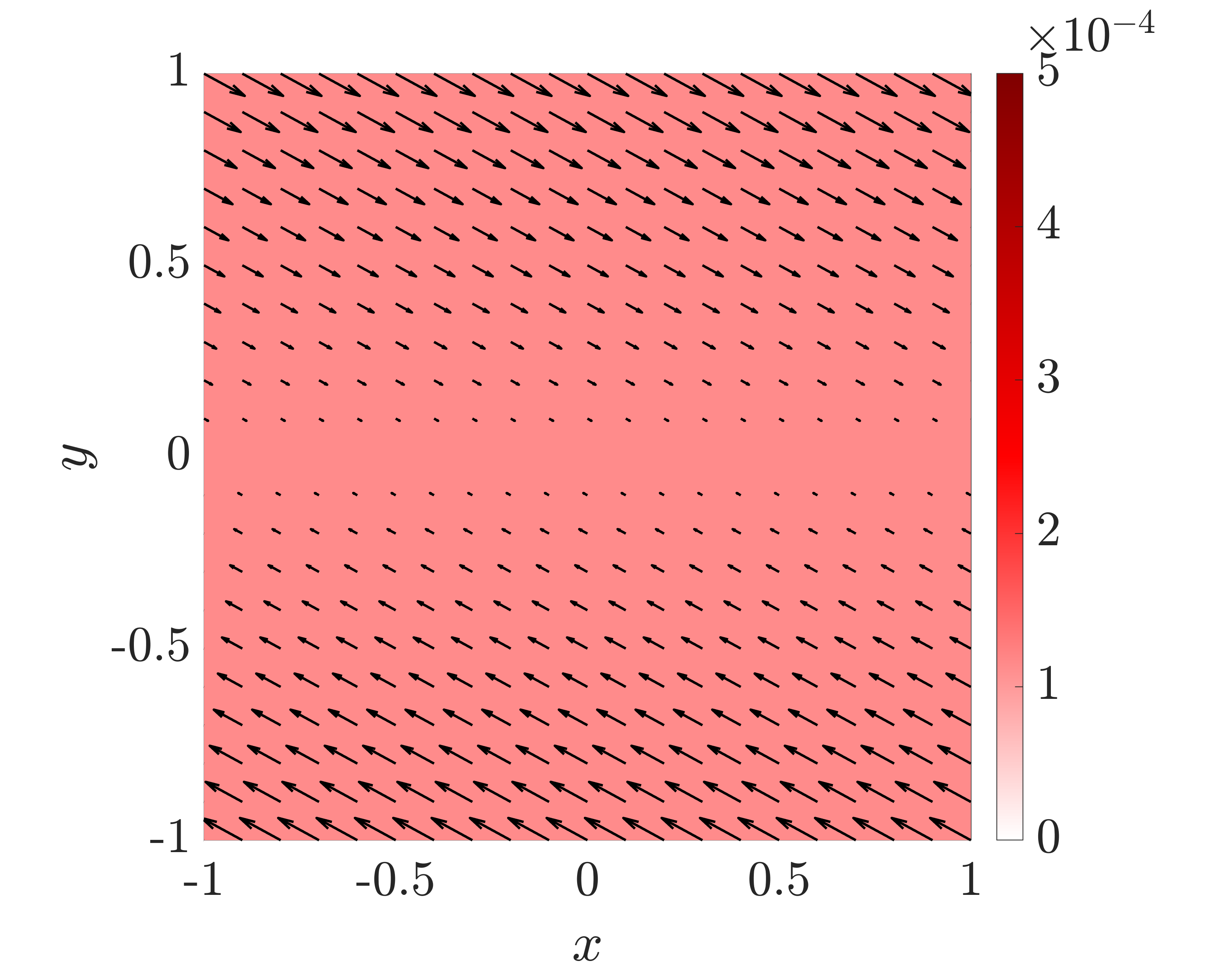}
	\caption{Velocity vector field $\boldsymbol{u}$ (arrow) and free surface height $h$ (contour) behaviour for Condition IV: (a) initial condition at $t=0$, (b) partial sum approximation based on ADM (with $N=6$) at $t=1$. Parameters used include $F=1$, $\bar{f}=0.5$, $D_0=0$, and $h_0=10^{-4}$ (Colour online).}
\label{fig:error_case_15}
\end{figure}

\section{Discussion}
\label{sec:discussion}

\par This work employs Adomian decomposition method (ADM) to the shallow water equations, where we made the following main contributions. First, we used these methods as reverse engineering mechanisms to develop theoretical connections between the ansatz formulations of previous works, such as \cite{thacker1981some}, \cite{shapiro1996nonlinear} and \cite{matskevich2019exact}, as well as develop a connection to the corresponding reduced systems of shallow water equations. Furthermore, we developed some novel families of closed-form solutions that respectively describe inertial oscillations and anticyclonic vortices with finite escape times over flat bottom topographies. We perform various numerical experiments against several cases that yielded relative errors between ${\rm O}\left(10^{-6}\right)$ and ${\rm O}\left(10^{-4}\right)$. Our numerical visualizations further demonstrate the validity of our approach, which illustrate the consistency with the dynamic behaviour for several scenarios while also preserving  the correlation between the physical parameters. 

\par Our study establishes the flexibility of these methods in terms of not only preserving the correlation of parameters with respect to the overall nonlinear physical behaviour but also alleviating the need to make restrictive assumptions like those based on the overall flow behaviour. Moreover, we illustrate that these techniques can be used to analytically deduce other aspects of shallow water phenomenon based on the characteristics of initial flows in which, to the best of our knowledge, this work is the first to explore these concepts. Therefore, some avenues of future work include extending these techniques to understand the implications of external forces such as the effects of bottom friction which are applicable to understanding various coastal effects such as impacts from tsunamis. Another area of research is extending this framework to analyse practical bottom topographies and shocks,  which will consider bottom terrains that extend beyond those of parabolic shapes.

\section*{Disclosure statement}

No potential conflict of interest was reported by the author(s).

\section*{Article Word Count}
3,710 words

\bibliographystyle{gGAF}
\bibliography{GAFD-2022-0015-main}

\begin{thebibliography}{42}
\providecommand{\natexlab}[1]{#1}

\bibitem[\protect\citeauthoryear{Adomian}{1990}]{adomian1990review}
Adomian, G., A review of the decomposition method and some recent results for
  nonlinear equations. {\itshape Math. Comput. Modell.}, 1990, \textbf{13},
  17--43.

\bibitem[\protect\citeauthoryear{Adomian}{2013}]{adomian2013solving}
Adomian, G., {\itshape Solving frontier problems of physics: the decomposition
  method},  Vol. 60,  2013 (Springer Science \& Business Media).

\bibitem[\protect\citeauthoryear{Ball}{1963}]{ball1963some}
Ball, F.K., Some general theorems concerning the finite motion of a shallow
  rotating liquid lying on a paraboloid. {\itshape J. Fluid Mech.}, 1963,
  \textbf{17}, 240--256.

\bibitem[\protect\citeauthoryear{Ball}{1964}]{ball1964exact}
Ball, F.K., An exact theory of simple finite shallow water oscillations on a
  rotating earth; in {\itshape Hydraul. Fluid Mech.}, 1964, pp. 293 -- 305.

\bibitem[\protect\citeauthoryear{Ball}{1965}]{ball1965effect}
Ball, F.K., The effect of rotation on the simpler modes of motion of a liquid
  in an elliptic paraboloid. {\itshape J. Fluid Mech.}, 1965, \textbf{22},
  529--545.

\bibitem[\protect\citeauthoryear{Bihlo {\itshape{et~al.}}}{2020}]{bihlo2020lie}
Bihlo, A., Poltavets, N. and Popovych, R.O., Lie symmetries of two-dimensional
  shallow water equations with variable bottom topography. {\itshape Chaos},
  2020, \textbf{30}, 073132.

\bibitem[\protect\citeauthoryear{Bila
  {\itshape{et~al.}}}{2006}]{mansfieldsymmetry}
Bila, N., Mansfield, E.L. and Clarkson, P.A., Symmetry group analysis of the
  shallow water and semi-geostrophic equations. {\itshape Q. J. Mech. Appl.
  Math.}, 2006, \textbf{59}, 95--123.

\bibitem[\protect\citeauthoryear{Bollermann
  {\itshape{et~al.}}}{2011}]{bollermann2011finite}
Bollermann, A., Noelle, S. and Luk{\'a}{\v{c}}ov{\'a}-Medvid’ov{\'a}, M.,
  Finite volume evolution {G}alerkin methods for the shallow water equations
  with dry beds. {\itshape Commun. Comput. Phys.}, 2011, \textbf{10}, 371--404.

\bibitem[\protect\citeauthoryear{Bristeau
  {\itshape{et~al.}}}{2021}]{bristeau2021analytical}
Bristeau, M.O., Di~Martino, B., Mangeney, A., Sainte-Marie, J. and Souillé,
  F., Some analytical solutions for validation of free surface flow
  computational codes. {\itshape J. Fluid Mech.}, 2021, \textbf{913}, A17.

\bibitem[\protect\citeauthoryear{Chesnokov}{2009}]{chesnokov2009symmetries}
Chesnokov, A.A., Symmetries and exact solutions of the rotating shallow-water
  equations. {\itshape Eur. J. Appl. Math.}, 2009, \textbf{20}, 461--477.

\bibitem[\protect\citeauthoryear{Chesnokov}{2011}]{chesnokov2011properties}
Chesnokov, A.A., Properties and exact solutions of the equations of motion of
  shallow water in a spinning paraboloid. {\itshape J. Appl. Math. Mech.},
  2011, \textbf{75}, 350--356.

\bibitem[\protect\citeauthoryear{Clark and Herron}{2013}]{clark2013improved}
Clark, A.D. and Herron, I.H., Improved bounds on linear instability of
  barotropic zonal flow within the shallow water equations. {\itshape Geophys.
  Astrophys. Fluid Dyn.}, 2013, \textbf{107}, 328--352.

\bibitem[\protect\citeauthoryear{Clarkson and Kruskal}{1989}]{clarkson1989new}
Clarkson, P.A. and Kruskal, M.D., New similarity reductions of the {B}oussinesq
  equation. {\itshape J. Math. Phys.}, 1989, \textbf{30}, 2201--2213.

\bibitem[\protect\citeauthoryear{Curr{\`o}}{1989}]{curro1989some}
Curr{\`o}, C., Some new exact solutions to the nonlinear shallow-water wave
  equations via group analysis. {\itshape Meccanica}, 1989, \textbf{24},
  26--35.

\bibitem[\protect\citeauthoryear{Cushman-Roisin}{1987}]{cushman1987exact}
Cushman-Roisin, B., Exact analytical solutions for elliptical vortices of the
  shallow-water equations. {\itshape Tellus A}, 1987, \textbf{39}, 235--244.

\bibitem[\protect\citeauthoryear{Cushman-Roisin
  {\itshape{et~al.}}}{1985}]{cushman1985oscillations}
Cushman-Roisin, B., Heil, W.H. and Nof, D., Oscillations and rotations of
  elliptical warm-core rings. {\itshape J. Geophys. Res.: Oceans}, 1985,
  \textbf{90}, 11756--11764.

\bibitem[\protect\citeauthoryear{Delestre
  {\itshape{et~al.}}}{2013}]{delestre2013swashes}
Delestre, O., Lucas, C., Ksinant, P.A., Darboux, F., Laguerre, C., Vo, T.N.T.,
  James, F. and Cordier, S., {SWASHES}: a compilation of shallow water analytic
  solutions for hydraulic and environmental studies. {\itshape Int. J. Numer.
  Methods Fluids}, 2013, \textbf{72}, 269--300.

\bibitem[\protect\citeauthoryear{Ern {\itshape{et~al.}}}{2008}]{ern2008well}
Ern, A., Piperno, S. and Djadel, K., A well-balanced {R}unge--{K}utta
  discontinuous {G}alerkin method for the shallow-water equations with flooding
  and drying. {\itshape Int. J. Numer. Methods Fluids}, 2008, \textbf{58},
  1--25.

\bibitem[\protect\citeauthoryear{Gallardo
  {\itshape{et~al.}}}{2007}]{gallardo2007well}
Gallardo, J.M., Par{\'e}s, C. and Castro, M., On a well-balanced high-order
  finite volume scheme for shallow water equations with topography and dry
  areas. {\itshape J. Comput. Phys.}, 2007, \textbf{227}, 574--601.

\bibitem[\protect\citeauthoryear{Iacono}{2005}]{iacono2005analytic}
Iacono, R., Analytic solutions to the shallow water equations. {\itshape Phys.
  Rev. E}, 2005, \textbf{72}, 017302.

\bibitem[\protect\citeauthoryear{Kafiabad
  {\itshape{et~al.}}}{2021}]{kafiabad2021interaction}
Kafiabad, H.A., Vanneste, J. and Young, W.R., Interaction of near-inertial
  waves with an anticyclonic vortex. {\itshape J. Phys. Oceanogr.}, 2021,
  \textbf{51}, 2035--2048.

\bibitem[\protect\citeauthoryear{Kesserwani and
  Liang}{2012}]{kesserwani2012locally}
Kesserwani, G. and Liang, Q., Locally limited and fully conserved {RKDG2}
  shallow water solutions with wetting and drying. {\itshape J. Sci. Comput.},
  2012, \textbf{50}, 120--144.

\bibitem[\protect\citeauthoryear{Levi {\itshape{et~al.}}}{1989}]{levi1989group}
Levi, D., Nucci, M., Rogers, C. and Winternitz, P., Group theoretical analysis
  of a rotating shallow liquid in a rigid container. {\itshape J. Phys. A:
  Math. Gen.}, 1989, \textbf{22}, 4743.

\bibitem[\protect\citeauthoryear{Li
  {\itshape{et~al.}}}{2017}]{li2017positivity}
Li, M., Guyenne, P., Li, F. and Xu, L., A positivity-preserving well-balanced
  central discontinuous Galerkin method for the nonlinear shallow water
  equations. {\itshape J. Sci. Comput.}, 2017, \textbf{71}, 994--1034.

\bibitem[\protect\citeauthoryear{Matskevich and
  Chubarov}{2019}]{matskevich2019exact}
Matskevich, N.A. and Chubarov, L.B., Exact solutions to shallow water equations
  for a water oscillation problem in an idealized basin and their use in
  verifying some numerical algorithms. {\itshape Numer. Anal. Appl.}, 2019,
  \textbf{12}, 234--250.

\bibitem[\protect\citeauthoryear{McKiver}{2020}]{mckiver2020balanced}
McKiver, W.J., Balanced ellipsoidal vortex at finite {R}ossby number. {\itshape
  Geophys. Astrophys. Fluid Dyn.}, 2020, \textbf{114}, 453--480.

\bibitem[\protect\citeauthoryear{Meleshko}{2020}]{meleshko2020complete}
Meleshko, S.V., Complete group classification of the two-dimensional shallow
  water equations with constant {C}oriolis parameter in {L}agrangian
  coordinates. {\itshape Commun. Nonlinear Sci. Numer. Simul.}, 2020,
  \textbf{89}, 105293.

\bibitem[\protect\citeauthoryear{Meleshko and
  Samatova}{2020}]{meleshko2020group}
Meleshko, S.V. and Samatova, N.F., Group classification of the two-dimensional
  shallow water equations with the beta-plane approximation of {C}oriolis
  parameter in {L}agrangian coordinates. {\itshape Commun. Nonlinear Sci.
  Numer. Simul.}, 2020, \textbf{90}, 105337.

\bibitem[\protect\citeauthoryear{Nikolos and
  Delis}{2009}]{nikolos2009unstructured}
Nikolos, I.K. and Delis, A.I., An unstructured node-centered finite volume
  scheme for shallow water flows with wet/dry fronts over complex topography.
  {\itshape Comput. Methods Appl. Mech. Eng.}, 2009, \textbf{198}, 3723--3750.

\bibitem[\protect\citeauthoryear{Pedlosky}{2013}]{pedlosky2013geophysical}
Pedlosky, J., {\itshape Geophysical fluid dynamics},  2013 (Springer Science \&
  Business Media).

\bibitem[\protect\citeauthoryear{Rogers}{1989{\natexlab{a}}}]{rogers1989generation}
Rogers, C., Generation of invariance theorems for nonlinear boundary-value
  problems in shallow water theory: an application of {MACSYMA}{\natexlab{a}};
  in {\itshape Numerical and Applied Mathematics, IMACS Meeting Proceedings,
  Paris}, 1989{\natexlab{a}}, pp. 69--74.

\bibitem[\protect\citeauthoryear{Rogers}{1989{\natexlab{b}}}]{rogers1989elliptic}
Rogers, C., Elliptic warm-core theory: The pulsrodon. {\itshape Phys. Lett. A},
  1989{\natexlab{b}}, \textbf{138}, 267--273.

\bibitem[\protect\citeauthoryear{Sachdev
  {\itshape{et~al.}}}{1996}]{sachdev1996regular}
Sachdev, P.L., Palaniappan, D. and Sarathy, R., Regular and chaotic flows in
  paraboloidal basins and eddies. {\itshape Chaos, Solitons Fractals}, 1996,
  \textbf{7}, 383--408.

\bibitem[\protect\citeauthoryear{Sampson
  {\itshape{et~al.}}}{2005}]{sampson2005moving}
Sampson, J., Easton, A. and Singh, M., Moving boundary shallow water flow above
  parabolic bottom topography. {\itshape ANZIAM J.}, 2005, \textbf{47},
  373--387.

\bibitem[\protect\citeauthoryear{Shapiro}{1996}]{shapiro1996nonlinear}
Shapiro, A., Nonlinear shallow-water oscillations in a parabolic channel: exact
  solutions and trajectory analyses. {\itshape J. Fluid Mech.}, 1996,
  \textbf{318}, 49--76.

\bibitem[\protect\citeauthoryear{Sun}{2016}]{sun2016high}
Sun, C., High-order exact solutions for pseudo-plane ideal flows. {\itshape
  Phys. Fluids}, 2016, \textbf{28}, 083602.

\bibitem[\protect\citeauthoryear{Thacker}{1977}]{thacker1977irregular}
Thacker, W.C., Irregular grid finite-difference techniques: simulations of
  oscillations in shallow circular basins. {\itshape J. Phys. Oceanogr.}, 1977,
  \textbf{7}, 284--292.

\bibitem[\protect\citeauthoryear{Thacker}{1981}]{thacker1981some}
Thacker, W.C., Some exact solutions to the nonlinear shallow-water wave
  equations. {\itshape J. Fluid Mech.}, 1981, \textbf{107}, 499--508.

\bibitem[\protect\citeauthoryear{Tsang and
  Dritschel}{2015}]{tsang2015ellipsoidal}
Tsang, Y.K. and Dritschel, D.G., Ellipsoidal vortices in rotating stratified
  fluids: beyond the quasi-geostrophic approximation. {\itshape J. Fluid
  Mech.}, 2015, \textbf{762}, 196--231.

\bibitem[\protect\citeauthoryear{Vallis}{2017}]{vallis2017atmospheric}
Vallis, G.K., {\itshape Atmospheric and oceanic fluid dynamics},  2017
  (Cambridge University Press).

\bibitem[\protect\citeauthoryear{Vallis}{2019}]{vallis2019essentials}
Vallis, G.K., {\itshape Essentials of Atmospheric and Oceanic Dynamics},  2019
  (Cambridge University Press).

\bibitem[\protect\citeauthoryear{Wintermeyer
  {\itshape{et~al.}}}{2018}]{wintermeyer2018entropy}
Wintermeyer, N., Winters, A.R., Gassner, G.J. and Warburton, T., An entropy
  stable discontinuous {G}alerkin method for the shallow water equations on
  curvilinear meshes with wet/dry fronts accelerated by {GPUs}. {\itshape J.
  Comput. Phys.}, 2018, \textbf{375}, 447--480.

\end{thebibliography}

\end{document}